\documentclass[preprint,10pt]{elsarticle}

\usepackage{hyperref}
\usepackage{xspace}
\usepackage{graphicx}
\usepackage{subfigure}
\usepackage{url}
\usepackage{algorithm}
\usepackage{algpseudocode}
\usepackage{xcolor}
\usepackage{amsmath}
\usepackage{amssymb}
\usepackage{amsthm}
\usepackage{enumitem}

\usepackage[a4paper,textwidth=150mm, textheight=240mm]{geometry}

\newtheorem{theorem}{Theorem}

\newtheorem{lemma}[theorem]{Lemma}

\theoremstyle{remark}
\newtheorem{remark}[theorem]{Remark}
\newtheorem{example}[theorem]{Example}

\journal{Journal of \LaTeX\ Templates}

\newcommand{\card}[1]{\ell(#1)}
\newcommand{\cardd}[1]{\ell'(#1)}

\algnewcommand{\IIf}[1]{\State\algorithmicif\ #1\ \algorithmicthen}
\algnewcommand{\EndIIf}{\unskip\ \algorithmicend\ \algorithmicif}

\newcommand{\comment}[1]{}

\newcommand{\ebox}{\hspace*{\fill}$\Box$}

\newcommand{\lsem}{\,[\kern-.15em[\,}
\newcommand{\rsem}{\,]\kern-.15em]\,}

\begin{document}

\begin{frontmatter}

\title{The Sum Composition Problem}

\author[unimiaddress]{Mario Pennacchioni}
\ead{mario.pennacchioni@usa.net}

\author[poliaddress]{Emanuele Munarini}
\ead{emanuele.munarini@polimi.it}

\author[unimiaddress]{Marco Mesiti\corref{mycorrespondingauthor}}
\ead{marco.mesiti@unimi.it}

%

\address[unimiaddress]{Dept. of Computer Science, University of Milano, Via Celoria, 18 -- Milano}
\address[poliaddress]{Dept. of Mathematics,  Politecnico of Milano, Piazza Leonardo da Vinci, 32 -- Milano}
\cortext[mycorrespondingauthor]{Corresponding author}

\begin{abstract}
 In this paper, we study the ``\emph{sum composition problem}''  between two lists $A$ and $B$ of positive integers.  We start by saying that $B$ is \emph{sum composition} of $A$  when there exists an ordered $m$-partition $[A_1,\ldots,A_m]$ of $A$  where $m$ is the length of $B$ and 
 the sum of each part $A_k$  is equal to the corresponding part of $B$.  
 Then, we consider the following two problems:
 $i)$ the \emph{exhaustive problem}, consisting in the generation of all partitions of $A$
     for which $B$ is sum composition of $A$, and
 $ii)$ the \emph{existential problem}, consisting in the verification of the existence of a partition of $A$
     for which $B$ is sum composition of $A$.
 Starting from some general properties of the sum compositions,
 we present a first algorithm solving the exhaustive problem
 and then a second algorithm solving the existential problem.
 We also provide proofs of correctness and experimental analysis
 for assessing the quality of the proposed solutions
 along with a comparison with related works.
\end{abstract}


\end{frontmatter}


\section{Introduction}

The ``{\em sum composition problem}'' between two lists
$A=[a_1, \ldots, a_n]$ and $B=[b_1, \ldots,b_m]$ of positive integers 
consists in the identification of a decomposition of the list $A$ into $m$ sub-lists $A_1,\ldots,A_m$  (where $m$ is the length of $B$) such that the sum of $A_k$  is equal to  $b_k \in B$ ($1 \leq k \leq m$). When this decomposition exists we say that $B$ is sum composition of $A$.

The sum composition problem has some analogies and shares the same complexity with the NP-hard {\em number partition problem} \cite{GJ79} where a list of integer $A$ is partitioned in $k$ partitions $P_1 \ldots P_k$ such that the sum of each part $P_j$ ($1 \leq j \leq k$) are as nearly equal as possible. However, there are several practical applications that can be modeled as  ``{\em sum composition problem}''. 
In the case of weighted graphs, the values in $A$ can correspond to the weights of the edges, and the values in $B$ are the aggregated values according to which the graph should be partitioned. In \cite{Chaimovich1993AFP}, the sum composition problem is shown as the problem of scheduling  independent tasks on uniform machines with different capabilities. In the case of databases, this problem has been encountered  by one of the authors when studying functional dependencies in a relational table. Specifically, when a functional dependency between an attribute $x$ and an attribute $y$ of a table $R$ exists, we can consider the list of the occurrences of values assumed by $x$ (denoted $A$) and the corresponding one assumed by $y$ (denoted by $B$). When $B$ is sum composition of $A$ it means that a functional dependency $x \rightarrow y$ can exist (this is a necessary condition but not sufficient).
The properties of sum composition can be exploited
for simplifying the checking of the existence of a functional dependency
by using the statistics associated with the table with no access to the single tuples.

In this paper, we start by interpreting this problem in terms of
a partial order relation between two integer partitions \cite{Ziegler}
and, then, we provide some general properties of such a relation.
There are two kinds of properties.
Those that guarantee the existence of a sum composition between two lists
and those that allow a simplification of the two lists
without affecting the property of being one sum composition of the other.
These properties are exploited in the realization of two algorithms.
The first one, named $\textsc{SumComp}$,
has the purpose to generate all possible decompositions of $A$
for which $B$ is sum composition of $A$.
The second one, named $\textsc{SumCompExist}$,
has the purpose to check the existence of at least one decomposition of $A$
for which the relation holds with $B$. 
In this last case, we are not interested in determining a decomposition but only to verify its existence.
The complexity of these algorithms is the same of the algorithms developed for the number partition problem because of the combinatorial explosion of the cases that need to be checked. However, the identified properties and the particular representation of the lists (in which distinct integers are reported along with their multiplicity) allows us to mitigate in some cases the  growth of the execution time curve.  


As  discussed in the related work section, 
the sum composition problem has received little attention from the research community.
However, its formulation can be of interest in many practical applications. 
Similar problems, presented in the literature for identifying optimal multi-way number partition \cite{SEK18}, the k-partitioning problem \cite{Chaimovich1993AFP}, and the sub-set sum problem \cite{Karp1972}, usually are devoted to identify a single correspondence between the two lists and sometimes they introduce restrictions on the values occurring in $A$ and $B$. Key characteristics of the proposed approach is the use of different properties of sum composition for reducing the controls in the algorithms and the adoption of a data structure for representing the partitions that reduce the number of combination of values to be considered. As shown in the experimental analysis section, these design strategies have positive effects on the running times of the algorithms (especially for the existential one) even if increasing the size of the partitions the execution times move quickly to an exponential rate.

The paper is structured as follows. Section~\ref{sec:pb} introduces the problem  and some notations. Then, Section~\ref{sec:properties} deals with the properties of sum composition. Section~\ref{sec:SCA} introduces the $\textsc{SumComp}$ algorithm and proves its  correctness.  Section~\ref{sec:existence} revises the previous algorithm for checking the  sum composition existence.  Section~\ref{sec:exp} presents the experimental results and Section~\ref{sec:rw} compares our results with related works. Finally, Section~\ref{sec:conclusion} draws the conclusions and  future research directions.

\section{Problem Definition}\label{sec:pb}

An \emph{integer partition} (or, simply, a \emph{partition}) of a 
positive 
integer $N$
is a list $ A = [a_1,a_2,\ldots,a_n] $
where each \emph{part} $ a_i $ is a positive integer,
$ a_1 \leq a_2 \leq \cdots \leq a_n $
and $ \sigma(A) = a_1 + a_2 + \cdots + a_n = N $.
A partition of $N$ can also be represented as
$ A = [ (a'_1,m_1), (a'_2,m_2), \ldots, (a'_h,m_h)]$,
where the numbers $ a'_1 < \cdots < a'_h $ are the distinct parts of $A$,
and $ m_1, m_2, \ldots, m_h $ are the respective multiplicities.
We write $\card{A}$ for the length of the list $A$,
i.e. for the number of parts of the partition $A$,
and $\cardd{A}$ for the number of distinct parts in $A$.
Usually, in the literature, an integer partition is written in decreasing order.
In this paper, however, we write an integer partition in increasing order because,
as we will see later, this notation is exploited in the formulation of the  algorithms.

\begin{example}
 The list $ A \!=\! [1,2,2,4,4,6] $ is an integer partition
 of $ N \!=\! \sigma(A) \!=\! 19 $ with $\card{A}\!=\!6$ parts and $\cardd{A}\!=\!4$ distinct parts.
 This partition can also be represented as $[(1,1),(2,2),(4,2),(6,1)]$. \ebox
\end{example}

A \emph{de\-com\-pos\-ition} of an integer partition $ A = [a_1,a_2,\ldots,a_n] $
is a list of integer partitions $ [A_1,A_2,\ldots,A_m] $ obtained as follows:
given a set partition $ \{I_1,I_2,\ldots,I_m\} $ of the index set $\{1,2,\ldots,n\} $ of $A$, $A_i$ is the integer partition whose parts are the $a_j$ with $j\in I_i$, for $ i = 1,2,\ldots,m$.

\begin{example}
 The partition $ A = [1,2,2,4,4,6] $ admits, for instance,
 the de\-com\-pos\-ition $ [A_1,A_2,A_3] $, where $A_1=[1,6] $, $A_2=[4] $ and $A_3=[2,2,4] $.\ebox
\end{example}

The set $ P_n $ of all integer partitions of $n$
can be ordered by refinement \cite[pp.\ 16, 1041]{Birkhoff} \cite{Ziegler} as follows:
given two integer partitions $ A = [a_1,a_2,\ldots,a_n] $ and $ B = [b_1,b_2,\ldots,b_m] $,
we have $ A \leq B $ whenever there exists a de\-com\-pos\-ition $ [A_1,A_2,\ldots,A_m] $ of $A$
such that $ \sigma(A_i) = b_i $, for every $ i = 1,2, \ldots, m $.
The set $ P_n $ is a very well known poset (partially ordered set)
whose topological properties have been studied in several papers
(e.g. see \cite{Ziegler} and the bibliography therein).

We will say that $B$ is \emph{sum composition} of $A$ whenever $A\leq B$.
Moreover, for simplicity, we will call \emph{$AB$-de\-com\-pos\-ition} or \emph{enabling de\-com\-pos\-ition}
any de\-com\-pos\-ition of $A$ for which $ A\leq B $.

\begin{example}\label{ex:es2}
 For the two integer partitions $ A = [1, 2, 2, 3, 4, 5] $ and $ B = [5, 5, 7] $,
 we have the following $AB$-de\-com\-pos\-itions:
 $ D_1 = [ [2, 3], [1, 4], [2, 5] ] $, $ D_2 = [[1, 4], [2, 3], [2, 5] ] $,
 $ D_3 = [[1, 2, 2],$ $[5], [3, 4] ] $, $ D_4 = [ [2, 3], [5], [1, 2, 4] ] $,
 $ D_5 = [[1, 4], [5], [2, 2, 3] ] $, $ D_6 = [[5], [1, 2, 2], [3, 4]] $,
 $ D_7 = [[5], [2, 3], [1, 2, 4] ] $ and $ D_8 = [[5], [1, 4], [2, 2, 3]] $. \ebox
\end{example}

In this way, the \emph{sum composition problem} can be formulated as follows:
\begin{itemize}[nosep]
 \item[(a)] generate all the $AB$-de\-com\-pos\-itions (exhaustive algorithm) between   the partitions  $A$ and $B$;
 \item[(b)] determine if $ A \leq B $ (existential algorithm).
\end{itemize}
The sum composition problem has also a simple puzzle interpretation \cite{JoniRota,Ziegler}.
Indeed, we have $ A \leq B $ whenever there is a tiling of the Ferrers diagram $\Phi_B$ of $B$
using the rectangles corresponding to the parts of $A$
(as in Figure~\ref{FigTiling}).
\begin{figure}[t]
 \centering
  \setlength{\unitlength}{5mm}
  \begin{picture}(16,5)(0,-1) 
   \put(0,0){
   \begin{picture}(7,4)
    \put(0,0){\line(1,0){2}}
    \put(0,1){\line(1,0){2}}
    \put(0,2){\line(1,0){4}}
    \put(0,3){\line(1,0){7}}
    \put(0,4){\line(1,0){7}}
    \put(0,0){\line(0,1){4}}
    \put(1,0){\line(0,1){4}}
    \put(2,0){\line(0,1){4}}
    \put(3,2){\line(0,1){2}}
    \put(4,2){\line(0,1){2}}
    \put(5,3){\line(0,1){1}}
    \put(6,3){\line(0,1){1}}
    \put(7,3){\line(0,1){1}}
    \put(3.5,-1){$(a)$}
   \end{picture}}
   \put(10,0){ 
   \begin{picture}(7,4)
   \newsavebox{\rectangle}
   \savebox{\rectangle}(0.8,0.8)[bl]
    {\put(0.1,0.1){\line(1,0){0.8}}
     \put(0.1,0.9){\line(1,0){0.8}}
     \put(0.1,0.1){\line(0,1){0.8}}
     \put(0.9,0.1){\line(0,1){0.8}}
    }
    \newsavebox{\rectangleb}
    \savebox{\rectangleb}(0.8,1.8)[bl]
    {\put(0.1,0.1){\line(1,0){1.8}}
     \put(0.1,0.9){\line(1,0){1.8}}
     \put(0.1,0.1){\line(0,1){0.8}}
     \put(1.9,0.1){\line(0,1){0.8}}
    }
    \newsavebox{\rectanglec}
    \savebox{\rectanglec}(0.8,2.8)[bl]
    {\put(0.1,0.1){\line(1,0){2.8}}
     \put(0.1,0.9){\line(1,0){2.8}}
     \put(0.1,0.1){\line(0,1){0.8}}
     \put(2.9,0.1){\line(0,1){0.8}}
    }
    \put(0,0){\line(1,0){2}}
    \put(0,1){\line(1,0){2}}
    \put(0,2){\line(1,0){4}}
    \put(0,3){\line(1,0){7}}
    \put(0,4){\line(1,0){7}}
    \put(0,0){\line(0,1){4}}
    \put(1,0){\line(0,1){4}}
    \put(2,0){\line(0,1){4}}
    \put(3,2){\line(0,1){2}}
    \put(4,2){\line(0,1){2}}
    \put(5,3){\line(0,1){1}}
    \put(6,3){\line(0,1){1}}
    \put(7,3){\line(0,1){1}}
    \put(3.5,-1){$(b)$}
    \put(0,0){\usebox{\rectangle}}
    \put(1,0){\usebox{\rectangle}}
    \put(0,1){\usebox{\rectangleb}}
    \put(0,2){\usebox{\rectangleb}}
    \put(2,2){\usebox{\rectangleb}}
    \put(0,3){\usebox{\rectanglec}}
    \put(3,3){\usebox{\rectanglec}}
    \put(6,3){\usebox{\rectangle}}
   \end{picture}}
  \end{picture}
 \caption{(a) Ferrers diagram $\Phi_B$ of the integer partition $ B = [2,2,4,7] $.
  (b) Tiling of $\Phi_B$ corresponding to the de\-com\-pos\-ition $ [ [1,1], [2], [2,2], [1,3,3] ] $
  of the integer partition $ A = [1,1,1,2,2,2,3,3] $.}
 \label{FigTiling}
\end{figure}
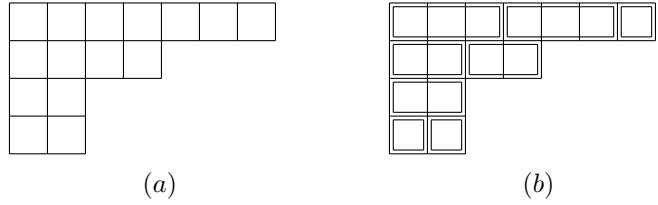

\begin{remark}
 Given an integer partition $ B = [b_1,b_2,\ldots,b_m] $,
 the number of all $AB$-de\-com\-pos\-itions with $ A \leq B $
 is $ p(b_1)p(b_2)\cdots p(b_m) $,
 where $ p(k) $ is the number of the integer partitions of $ k $ \cite[p. 307]{Comtet}.
 Similarly, given an integer partition $ A = [a_1,\ldots,a_n] $,
 the number of all $AB$-de\-com\-pos\-itions with $ A \leq B $ is at most $ b_n $,
 where $ b_n $ is the Bell number of order $n$ \cite[p. 210]{Comtet},
 i.e. of the number of set partitions of an $n$-set.
 Finally, given two integer partitions $ A = [a_1,\ldots,a_n] $
 and $ B = [b_1,b_2,\ldots,b_{m-1},b_m] $ with $ A \leq B $,
 the number $ d(A,B ) $ of $AB$-de\-com\-pos\-itions satisfy the bounds
 $ 1 \leq d(A,B) \leq p(b_1) p(b_2)\cdots p(b_{m-1}) $.
 Clearly, if $ A = [1,\ldots,1] $ has length $m$ and $ B = [m] $,
 we have $ A \leq B $ and $ d(A,B) = 1 $.
 Moreover, for the partitions $ A = [1,1,1,1,1,2,2,3] $ and $ B = [2,3,7] $,
 we have $ A \leq B $ and $ d(A,B) = 6 = p(2)p(3) $.
\end{remark}

In the rest of the paper, we will use the following operations on integer partitions.
The union of two partitions $A$ and $B$ is the partition $ A\cup B$
whose parts are those of $A$ and $B$, arranged in increasing order.
Similarly, the intersection of $A$ and $B$ is the partition $ A\cap B$
whose parts are those common to $A$ and $B$, arranged in increasing order.
$A$ is a subpartition of $ B $, i.e. $ A \subseteq B $,
when every part of $A$ is also a part of $B$.
If $A$ is a subpartition of $B$, then we write $ B\setminus A$ for the partition
whose parts are those of $B$ taken off those of $A$.
Finally, we write $ a \in A $ when $a$ is a part of $ A $.
In this way, we have $ \sigma(A\cup B) = \sigma(A) + \sigma(B) $ and
$ \sigma(B\setminus A) = \sigma(B) - \sigma(A) $.

\begin{example}
 Given  $ A = [1,1,3,4] $ and $ B = [1,2,4,4,5] $,
 we have $ A \cup B = [1,1,1,2,3,4,4,4,5] $ and $ A\! \cap\! B \!=\! [1,4] $.
 Given  $ A \!=\! [1,1,3,4] $ and $B\! = \![1,1,2,3,3,4,5] $,
 we have $ A\! \subseteq \!B $ and $ B\! \setminus\! A \!=\! [2,3,5]$. \ebox
\end{example}

\section{Properties of  Sum Composition}\label{sec:properties}

In this section, we present some properties of the sum composition relation
that will be exploited   in the exhaustive and existential algorithms  presented in the next sections.

\subsection{Basic Properties}

We start by listing 
some basic necessary conditions for the existence of a sum composition.

\begin{lemma}\label{lemma:1}
 Let $A=[a_1,\ldots,a_n]$ and $B=[b_1,\ldots,b_m]$ be two integer partitions such that $A \leq B$,
 and let $ [A_1,\ldots,A_m] $ be an $AB$-de\-com\-pos\-ition.
 Then, we have
 \begin{enumerate}[nosep]
  \item $n \geq m$
  \item $a_1 \leq b_1$ (and, more generally, $a_1 \leq b_i$ for every $i$)
  \item $a_n \leq b_m$ (and, more generally, $a_i \leq b_m$ for every $i$)
  \item if $ a_i \in A_j $, then $ a_i \leq b_j $
  \item $ \sigma(A) = \sigma(B) $.
 \end{enumerate}
\end{lemma}
\begin{proof}
 These properties derive directly from the definition of sum composition.
 \begin{enumerate}[nosep]
  \item Immediate from the definition of the order relation $\leq$.
  \item The element $a_1$ is the minimum value appearing in $A$.
   Then $A_1$ contains at least one part greater or equal to $a_1$
   and consequently $ b_1 = \sigma(A_1) \geq a_1 $.
  \item The element $a_n$ is the maximum value appearing in $A$
   and appears in some $A_i$.
   Similarly, the element $b_m$ is the maximum value appearing in $B$.
   So, we have $ b_m \geq b_i = \sigma(A_i) \geq a_n $.
  \item If $ a_i \in A_j $, then $ b_j = \sigma(A_j) = \cdots + a_i + \cdots \geq a_i $.
  \item We have $ \sigma(A) = \sigma(A_1) + \cdots + \sigma(A_m) = b_1 + \cdots + b_m = \sigma(B) $.
 \end{enumerate}

\vspace*{-.55cm}
\end{proof}

We have also the following useful condition.
\begin{theorem}\label{theo:3}  
 Let $ A=[a_1,\ldots,a_n] $ and $B=[b_1,\ldots,b_m]$ be two integer partitions
 s.t. $A \leq B$.
 If the intersection $ A \cap B $ has $ k $ parts, then $ n \geq 2m-k $.
\end{theorem}
\begin{proof}
 Since $ A \leq B $,  an $AB$-de\-com\-pos\-ition $ [A_1,\ldots,A_m] $ exists.
 Since $A$ and $B$ have $k$ common parts,
 there are at most $k$ partitions $A_i$ with only one part
 and all the remaining partitions contains at least two parts.
 Consequently $ n \geq k + 2(m-k) = 2m -k $.
\end{proof}

\begin{remark}
 By Theorem~\ref{theo:3},
 if $ A=[a_1,\ldots,a_n] $ and $B=[b_1,\ldots,b_m]$ are two integer partitions
 with $ A \cap B $ with $ k $ parts and $ n \leq 2m-k-1 $,
 then $B$ can not be a sum composition of $A$.
\end{remark}

The following lemma will be used in the subsequent sections to write the algorithms.
It allows you to proceed iteratively on the elements of B, by considering them one by one.
The validity of the hypothesis of this lemma is guaranteed at each step.
\begin{theorem}\label{lemma:sumnm1}
 Let $A=[a_1,\ldots,a_n]$ and $B=[b_1,\ldots,b_m]$ be two integer partitions s.t. $\sigma(A) = \sigma(B)$.
 If $B'\!=\!B\! \setminus\! [b_m]$ is sum composition of some subpartition $A'$ of $A$,
 then $B$ is sum composition of $A$.
\end{theorem}
\begin{proof}
 Since $B'$ is sum composition of $A'$, there exists an $A'B'$-de\-com\-pos\-ition $ [A_1,\ldots,A_{m-1}] $.
 Let $ A_m = A \setminus A' $. Then
 $ \sigma(A_m) = \sigma(A) - \sigma(A') = \sigma(A) - ( \sigma(A_1) + \cdots + \sigma(A_{m-1}) )
 = \sigma(A) - ( b_1 + \cdots + b_{m-1} ) = \sigma(A) - ( \sigma(B) - b_m ) = b_m $.
 Hence, $ [A_1,\ldots,A_{m-1},A_m] $ is an $AB$-de\-com\-pos\-ition.
\end{proof}

Let $A^{(k)}$ be the partition whose parts are the parts of a partition $A=[a_1,\ldots,a_n]$ less or equal to $k$ and
let $A^{[k]}$ be the partition whose parts are the parts of $A$ grater or equal to $k$.

\begin{theorem}\label{theo:aGb} 
 Let $A=[a_1,\ldots,a_n]$ and $B=[b_1,\ldots,b_m]$ be two integer partitions
 s.t. $B$ is sum composition of $A$. For every $ k = 1, \ldots, b_m $, we have
 $$
  \sigma(A^{(k)}) \geq \sigma(B^{(k)})
  \qquad\text{and}\qquad
  \sigma(A^{[k]}) \leq \sigma(B^{[k]})\, .
 $$
\end{theorem}
\begin{proof}
 Since $B$ is sum composition of $A$,
 there exists an $AB$-de\-com\-pos\-ition $ [A_1,\ldots,A_m] $.
 If $ B^{(k)} = [b_1,\ldots,b_h] $,
 then $ \sigma(A_1) = b_1 \leq k $, \ldots, $ \sigma(A_h) = b_h \leq k $.
 Hence, all elements of $A_1$, \ldots, $A_h$ are less or equal to $k$,
 and consequently $A_1 \cup \cdots \cup A_h \subseteq A^{(k)} $.
 Then $ \sigma(A^{(k)}) \geq \sigma(A_1) + \cdots + \sigma(A_h) = b_1 + \cdots + b_h = \sigma(B^{(k)}) $.

 If $k=1$, we have $ A^{[1]} = A $ and $ B^{[1]} = B $.
 So $ \sigma(A^{[1]}) = \sigma(B^{[1]}) $, being $ \sigma(A) = \sigma(B) $.
 If $k>1$, then $ A^{[k]} = A \setminus A^{(k-1)} $ and $B^{[k]} = B \setminus B^{(k-1)} $.
 So $ \sigma(A^{[k]}) = \sigma(A) - \sigma(A^{(k-1)}) \leq \sigma(B) - \sigma(B^{(k-1)}) = \sigma(B^{[k]}) $.
\end{proof}

We have also the following divisibility property
(saying, on the other hand, that the sum composition property is preserved by multiplication).
\begin{theorem}\label{theo:d1}
 Let $A$ 
 and $B$ 
 be two integer partitions
 s.t. $B$ is sum composition of $A$.
 If there exists an integer $d>1$ dividing all parts of $A$,
 then $d$ divides all parts of $B$.
\end{theorem}
\begin{proof}
 Since $A\leq B$, an $AB$-de\-com\-pos\-ition $ [A_1,\ldots,A_m] $ exists.
 Moreover, by hypothesis, every $a_i$ is divisible by $d$,
 we have at once that each $ b_i = \sigma(A_i) $ is divisible by $d$.
\end{proof}

\begin{example}
 Consider the integer partitions $ A = [50,100,100,200,250,300] $ and $ B = [300,300,400] $.
 Every element of $A$ and $B$ is divisible by $50$.
 So, by dividing by 50, we have the partitions $ A' = [1,2,2,4,5,6] $ and $ B' = [6,6,8] $.
 Since $ [[6],[1,5],[2,2,4]] $ is an $A'B'$-de\-com\-pos\-ition, then $A'\leq B'$,
 and consequently $A\leq B$.
 Notice that to obtain an $AB$-de\-com\-pos\-ition it is sufficient to multiply all parts of
 an $A'B'$-de\-com\-pos\-ition by $50$. \ebox
\end{example}

Finally, we have that the sum composition property is preserved by the union of partitions.
More precisely, we have the following result.
\begin{theorem}\label{theo:Union}
 Let $A'=[a'_1,\ldots,a'_{n'}]$, $B'=[b'_1,\ldots,b'_{m'}]$,
 $A''=[a''_1,\ldots,a''_{n''}]$ and $B''=[b''_1,\ldots,b''_{m''}]$
 be four integer partitions s.t. $B'$ is sum composition of $A'$ and $B''$ is sum composition of $A''$.
 Then $B=B' \cup B''$ is sum composition of $A=A' \cup A''$.
\end{theorem}
\begin{proof}
 Since $A'\leq B'$,  an $A'B'$-de\-com\-pos\-ition $ [A_1',\ldots,A_{m'}'] $ exists.
 Similarly, since $A''\leq B''$,  an $A''B''$-de\-com\-pos\-ition $ [A_1'',\ldots,$ $A_{m''}''] $ exists.
 Hence, by properly merging $ [A_1',\ldots,A_{m'}'] $ and $ [A_1'',\ldots,$ $A_{m''}''] $,
 we obtain a de\-com\-pos\-ition of $A' \cup A''$ for which $ A' \cup A'' \leq B' \cup B'' $.
\end{proof}

\begin{example}
 Consider the partitions $ A' = [1,1,2,3,4] $ and $ B' = [3,3,5] $
 with the $A'B'$-de\-com\-pos\-ition $ [[1,2],[3],[1,4]]$,
 and the partitions $ A'' = [1,1,2,3,5] $ and $ B'' = [1,5,6] $
 with the $A''B''$-de\-com\-pos\-ition $ [[1],[1,3],[1,6]]$.
 Then $ A' \cup A'' = [1,1,1,1,2,2,3,3,4,5] $ and $ B' \cup B'' = [1,3,3,5,5,6] $,
 and $ [[1],[1,2],[3],[1,4],[2,3],$ $[1,6]]$ is a de\-com\-pos\-ition of
 $ A' \cup A'' $ for which $ A' \cup A'' \leq B' \cup B'' $. \ebox
\end{example}

\begin{remark}
 Given two partitions $A$ and $B$ with $A\leq B $,
 and given two subpartions $ A'\subseteq A $ and $ B' \subseteq B $ with $A'\leq B'$,
 for the complementary subpartitions $ A'' = A\setminus A' $ and $ B'' = B\setminus B' $
 it is not said that $ A'' \leq B'' $.
 Consider, for instance, the partitions $ A = [1,1,1,2,2,2,3] $ and $ B = [2,2,3,5] $
 and the partitions $ A' = [1,1,2,2,2] \subseteq A $ and $ B' = [3,5] \subseteq B $.
 Then $ A \leq B $ and $ A' \leq B' $.
 However, we have $ A'' = [1,3] $ and $ B'' = [2,2] $, and $ A'' \not\leq B'' $.
\end{remark}

\subsection{Reduction Properties}\label{subsection:existence}

In this section, we  present some properties
that can be used for eliminating values from the partitions involved in a sum composition,
without altering the property of being a sum composition.
These properties can positively effecting the execution time of the  existence algorithm.

\begin{theorem}\label{theo:theo1}
 Let $A=[a_1,\ldots,a_n]$ and $B=[b_1,\ldots,b_m]$ be two integer partitions
 s.t. $A \leq B$.
 If  two indices $ i $ and $j$ exist s.t. $a_i=b_j$,
 then $B'=[b_1,\ldots,b_{j-1},$ $b_{j+1},\ldots, b_m]$  is sum composition of $A'=[a_1,\ldots,a_{i-1}, a_{i+1},\ldots, a_n]$.
\end{theorem}
\begin{proof}
 Since $A\leq B$,  an $AB$-de\-com\-pos\-ition $ [A_1,\ldots,A_m] $ exists,
 where, by hypothesis, $ a_i = b_j = \sigma(A_j) $.
 We have the following two cases.
 \begin{enumerate}[nosep]
  \item If $ a_i \in A_j $, then $ A_j = [a_i] $.
   Hence $[A_1,\ldots,A_{j-1},A_{j+1}, \ldots,A_m] $ is an $A'B'$-de\-com\-pos\-ition.
  \item If $ a_i \not\in A_j $, then there exists an index $ k \ne j $ such that $ a_i \in A_k $.
   Let $ A'_k = (A_j\cup A_k)\setminus[a_i] $.
   Then, replacing $A_k$ by $A'_k$ in $ [A_1,\ldots,A_{j-1},A_{j+1},\ldots,A_m] $,
   we obtain a de\-com\-pos\-ition of $A'$ such that $ A' \leq B' $.
   Indeed, we have $ \sigma(A'_k) = \sigma(A_j) + \sigma(A_k) - a_i = b_j + b_k - a_i = b_k $.
 \end{enumerate}

\vspace*{-.5cm}
\end{proof}

Theorem~\ref{theo:theo1} can be easily generalized as follows.
\begin{theorem}\label{cor1}
 Let $A$ and $B$ be two integer partitions
 s.t. $B$ is sum composition of $A$.
 If $C$ is a subpartion of $A$ and $B$,
 then $B \setminus C$ is sum composition of $A \setminus C$.
\end{theorem}
\begin{proof}
 Suppose that $ C = [c_1,\ldots,c_k] $.
 Since $A$ and $B$ have a common part $c_1$,
 then $B'=B\setminus[c_1]$ is a sum composition of $A'=A\setminus[c_1]$,
 by Theorem~\ref{theo:theo1}.
 Since $A'$ and $B'$ have a common part $c_2$,
 then $B''=B'\setminus[c_2]$ is a sum composition of $A''=A'\setminus[c_2]$,
 always by Theorem~\ref{theo:theo1}.
 Continuing in this way, at the end we have that
 $B \setminus C$ is sum composition of $A \setminus C$.
\end{proof}

We also have the following result.
\begin{theorem}\label{cor11}
 If $A$ and $B$ are two partitions s.t. $A \leq B$,
 then, for every partition $C$, $A \cup C \leq B \cup C$.
\end{theorem}
\begin{proof}
 By hypothesis, we have $ A \leq B $, and, clearly, $ C \leq C $.
 Hence, by Theorem~\ref{theo:Union}, we have $ A\cup C \leq B \cup C $.
\end{proof}

Theorems~\ref{cor1} and~\ref{cor11} immediately implies the following property.
\begin{theorem}
 If $A$, $B$ and $C$ are three integer partitions,
 then $A\leq B$ if and only if $A \cup C \leq B \cup C$.
\end{theorem}

\begin{remark}
 Equivalently, if $A$ and $B$ are two integer partitions
 and $C$ is a subpartion of $A$ and $B$,
 then $A\leq B$ if and only if $A \setminus C \leq B \setminus C$.
\end{remark}

\newcommand{\com}[1]{\Comment{\color{gray}#1\color{black}}}

Finally, we have the following further simplification property
which will be used later to reduce the running time of the existence algorithm.

\begin{theorem}\label{theo:simpl2}
 Let $A=[a_1,\ldots,a_n]$ and $B=[b_1,\ldots,b_{m-1},b_m]$ be two integer partitions
 s.t. $ b_{m-1} < a_n < b_m $. Then $B$ is sum composition of $A$ if and only if
 $B'=[b_1,\ldots,b_{m-1}]\cup[b_m-a_n]$ is sum composition of $A'=[a_1,\ldots,a_{n-1}]$.
\end{theorem}
\begin{proof}
 If $A\leq B$, then  an $AB$-de\-com\-pos\-ition $ [A_1,\ldots,A_m]$ exists.
 Suppose that $ a_n \in A_k $.
 Then, by item 4 of Lemma~\ref{lemma:1}, we have $ a_n \leq b_k $.
 Hence, being $b_{m-1} < a_n < b_m$, we have $ k = m $, that is $ a_n \in A_m $.
 Let $ A'_m = A_m \setminus [a_n] $.
 Then $ [A_1,\ldots,A_{m-1},A_m'] $ is a de\-com\-pos\-ition of $A'$
 for which $A'\leq B'$ (indeed,  $\sigma(A'_m) = \sigma(A_m) - a_n = b_m - a_n $).

 If $A'\leq B'$, then  an $A'B'$-de\-com\-pos\-ition $ [A'_1,\ldots,A'_m]$ exists.
 If $A'_k$ corresponds to $b_m-a_n \in B'$,
 then $ [A'_1,\ldots,A'_{k-1}, A'_k \cup [a_n], A'_{k+1},\ldots,A'_m] $
 is an $AB$-de\-com\-pos\-ition.
 In fact, for all $A'_i$ with $i\ne k$ the values are the same,
 and for $A_k=A'_k \cup [a_n]$ we have $\sigma(A_k)\!=\!\sigma(A'_k)\!+\!a_n\!=\!b_m\!-\!a_n\!+\!a_n\!=\!b_m$.
\end{proof}

\begin{example}
 Consider  $ A = [1,1,2,2,4] $ and $ B = [1,3,6] $ with $ A\leq B $.
 Since $ 3 < 4 < 6 $, for the partitions $ A' = [1,1,2,2] $ and $ B' = [1,2,3] $ also holds $ A'\leq B' $. The viceversa is true as well. \ebox
\end{example}

\comment
{ ***INIZIO COMMENTO
\begin{theorem}\label{theo:simpl2}
 Let $A=[a_1,\ldots,a_n]$ and $B=[b_1,\ldots,b_m]$ be two integer partitions
 s.t. $B$ is sum composition of $A$.
 If $b_m > a_n > b_{m-1}$,
 then  $B'=[b_1,\ldots,b_{m-1}]\cup[b_m-a_n]$ is sum composition of $A'=[a_1,\ldots,a_{n-1}]$.
\end{theorem}
\begin{proof}
 Since $A\leq B$, there exists an $AB$-de\-com\-pos\-ition $ [A_1,\ldots,A_m] $.
 Suppose that $ a_n \in A_k $.
 Then, by item 4 of Lemma~\ref{lemma:1}, we have $ a_n \leq b_k $.
 Hence, being $b_m > a_n > b_{m-1}$, we have $ k = m $, that is $ a_n \in A_m $.
 Let $ A'_m = A_m \setminus [a_n] $.
 Then $ [A_1,\ldots,A_{m-1}A_m'] $ is a de\-com\-pos\-ition of $A'$
 for which $A'\leq B'$.
 Indeed, we have $ \sigma(A'_m) = \sigma(A_m) - a_n = b_m - a_n $.
\end{proof}

\begin{example}
 Consider the integer partitions $ A = [1,1,2,2,4] $ and $ B = [1,3,6] $ with $ A\leq B $.
 Since $ 6 > 4 > 3 $, we can apply Theorem~\ref{theo:simpl2}
 where $ A' = [1,1,2,2] $ and $ B' = [1,3,3] $ with $ A'\leq B' $. \ebox
\end{example}

*** FINE COMMENTO}

\section{The Exhaustive Sum Composition Algorithm}\label{sec:SCA}

In this section, we present an algorithm for generating all $AB$-de\-com\-pos\-itions
of two integer partitions $A=[a_1,\ldots,a_n]$ and $B=[b_1,\ldots,b_m]$.
To develop such an algorithm,
we represent an $AB$-de\-com\-pos\-ition $[A_1,\ldots,A_m]$ for $B$
as a list $[(b_1, A_1), \ldots, (b_m, A_m)]$,
where $\sigma(A_j) = b_j$ for every $j = 1, \ldots, m$ to better clarify the association between the parts in $B$ and the decomposition of $A$.

\begin{algorithm}[t]
\begin{footnotesize}
\begin{algorithmic}[1]
\Function {SumCompAux}{$b, A, p$}
    \IIf {$b < a_p$ {\bf or} $p > \cardd{A}$}
 	   	{\bf return} $\emptyset$
 	 \EndIIf
\State $quotient \leftarrow min(b~{\tt div}~a_p, m_p)$ \com{\small Where  $(a_p, m_p)$ is the $p^{th}$ element of  $A$}

   \State $Result \leftarrow \emptyset$	
  \For {$i \leftarrow 0 \ldots quotient$}
    \If {$i=quotient$ {\bf and} $quotient \cdot a_p  =b$}
    \State     $Result = Result \cup \{\{(a_p,quotient)\}\}$
    \Else
	  \State $\{A_1,\ldots,A_h\} \leftarrow \textsc{SumCompAux}(b - i  a_p, A, p+1)$
    \If {$h>0$}
    \If {$i=0$}
    \State  $Result = Result \cup \{A_1,\ldots,A_h\}$
    \Else
    \State {$Result = Result \cup \{A_1 \cup \{(a_p,i)\},\ldots,A_h\cup \{(a_p,i)\}\}$}
    \EndIf
    \EndIf
    \EndIf

  \EndFor
\State {\bf return} $Result$
\EndFunction
\end{algorithmic}
\end{footnotesize}
\end{algorithm}

\subsection{The  $\textsc{SumCompAux}$ function}

Before presenting the general algorithm,
we introduce the recursive function $\textsc{SumCompAux}$ having three parameters:
an integer partition $A$, an integer $b$ and a position $p$.
The purpose of this function is to identify all partitions of the elements of $A$
whose positions are greater (or equal to) $p$ such that their sum is equal to $b$.
Of course, when $p=1$ the recursive function $\textsc{SumCompAux}$ identifies all the subpartitions of $A$ such that their sum is  $b$.

\begin{example}\label{ex:esSumAux} Consider   $A=\{(50,1),(100,2),(200,1),(250,1),(300,1)\}$ and $B = [300, 300, 400]$ of Example~\ref{ex:es2}.  When  $\textsc{SumCompAux}$ is called on $A$, a part  $b \in B$,  and position $1$, we expect that it provides all the subpartitions  $A_1,\ldots, A_h$ of $A$ s.t. $\sigma(A_1)=\ldots=\sigma(A_h)=b$.  Specifically:
 \begin{itemize}[nosep]
  \item $\textsc{SumCompAux}(300,A,1)\!=\!\{\!\{(50,1),(250,1)\},\!\{(100,1),(200,1)\},\!\{(300,1)\}\!\}$,
  \item $\textsc{SumCompAux}(400,A,1)=\{\{(50,1),(100,1),(250,1)\},\{(100,2),(200,1)\},\{(100,1),(300,1)\}\}$.
 \end{itemize}
 Note that $\textsc{SumCompAux}(300,A,4)=\{\{(300,1)\}\}$. That is, $\{(250,1),(300,1)\}$ are the only parts of $A$ that are considered for determining the value $300$. \ebox
\end{example}

At each step of the recursion, 
the function evaluates all the  subpartitions of $A$ that can be generated by taking into account the element at position $p$ (denoted $a_p$). If $a_p$ is greater than $b$ (since  $A$ is ordered),  no further subpartitions can be determined and the empty set is returned (see Item 3 of Lemma~\ref{lemma:1}). The same conclusion is obtained when the position $p$ is greater than $\card{A}$. 
When these cases are not met, it means that $b \geq a_p$. Therefore, we can consider the possibility to use $a_p$ (or not) in the identification of the subpartitions whose sum is equal to $b$.
For this purpose we determine $quotient$ as the minimum between the possible multiplicity of $a_p$ (denoted with $m_p$) and the integer division of $b$ by $a_p$. It represents the maximum number of times that $a_p$ can be summed to obtain the value $b$.
When the current element is not considered for determining the subpartitions of $A$ whose sum is $b$, the function will return as output the subpartitions that can be determining starting from the element of $A$ of position $p+1$ for the same element $b$.
When the current element is taken $i>0$ times, the recursive call is invoked on the value ($b-i  a_p$) starting from the next position. Therefore, it determines the subpartitions $A_1, \ldots, A_h$ of $A$ whose sum is ($b-i  a_p$) and will return as output these sets in which element $a_p$ is taken $i$ times. When $ quotient \cdot a_p =b$ there is no needs to proceeds with further recursive calls, and the pair $(a_p, quotient)$ can be directly returned. Note that, when a call to the recursive function returns the empty set, it means that along this path no results can be obtained.

\begin{figure}[t]
\centering \includegraphics[scale=.5]{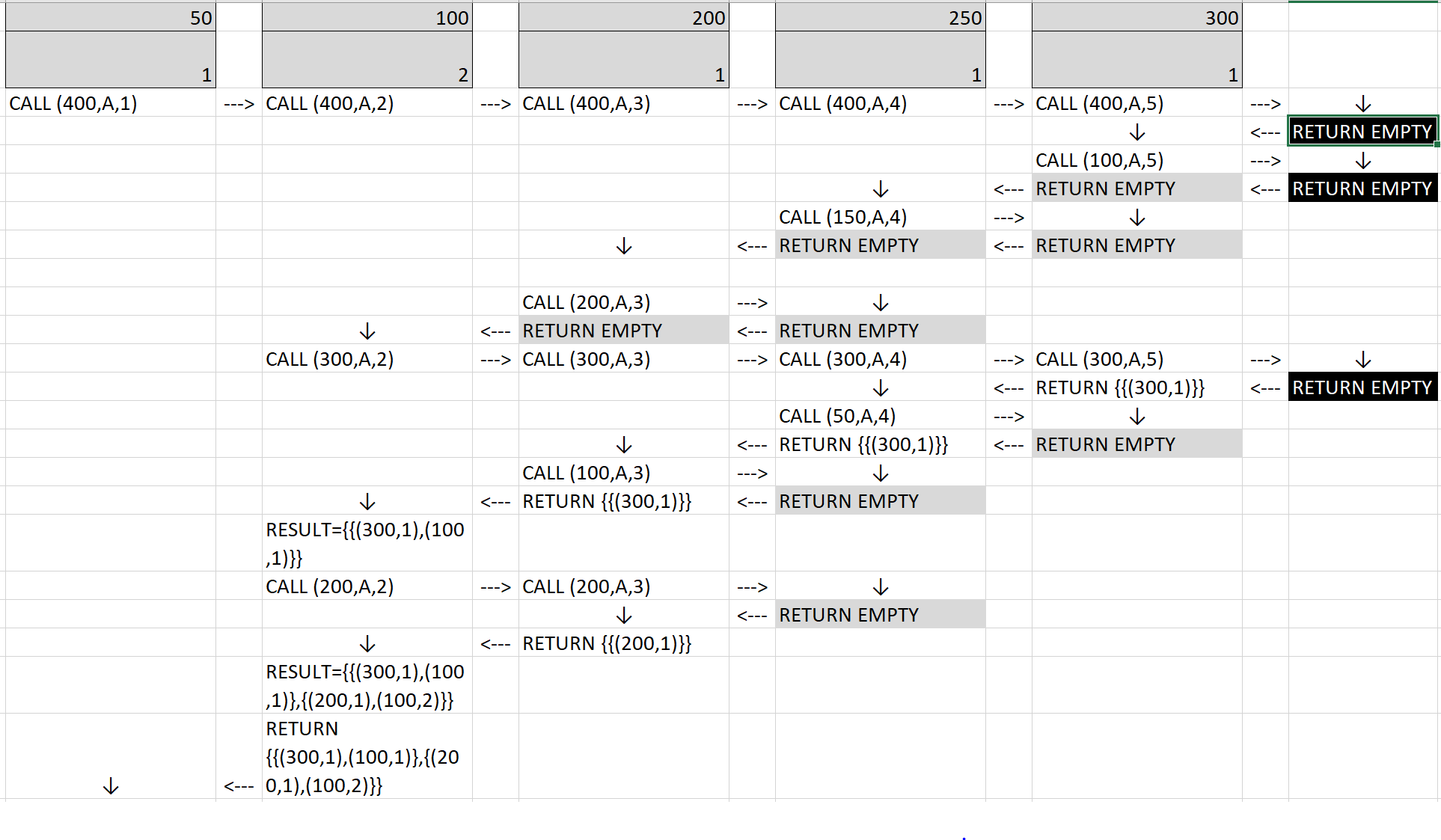}
\caption{Trace of the execution of function $\textsc{sumCompAux}$}\label{fig:sumCompAux}
\end{figure}

\begin{example} Figure~\ref{fig:sumCompAux} reports an excerpt of the trace of the calls to function $\textsc{sumCompAux}$. 
In the top of the figure the parts of $A$ with their repetitions are shown. The string ``return empty" with the black background represents the situation in which we have reached the end of the partition without finding sub-lists whose sum is $b$, whereas 
the gray background means that it is useless to proceed because $a_p > b$.
Each column reports all the calls for the same element of $A$.

The first line reports the calls in which the current element of $A$ is not selected. This leads to reach the end of partition $A$ without finding any sub-list of $A$ whose sum is $400$. In returning from the recursive call, the function tries to create a single instance of $300$, and to identify instances of $A$ of position $6$ whose sum is $100$. However, since we have reached the end of $A$, the empty set is returned. The same empty sets are returned in the call backs of the function until the second column ({\tt call(300,A,2)}) is reached. For this case, indeed, we are able to identify a partition  $Result=\{\{(300,1),(100,1)\}\}$.  $Result=\{\{(200,1),(100,2)\}\}$ is another subset that is obtained for {\tt call(200,A,2)}. Since we have reached to maximum number of repetitions of  $100$ in $A$, the last two results are collected in a partition and returned.

The algorithm proceeds in identifying new results when the part $50$ in $A$ is selected once. We do not further present the other calls, that are analogous to those we have now discussed. \ebox
\end{example}

The following theorem shows that  $\textsc{SumCompAux}$  invoked on $A$, $b$ and position $p$
provides all the partitions $A_j$ s.t. $ \sigma(A_j) = b $.
This theorem proves that this is true for all $a_i \in A$ where $i\geq p$.

\begin{theorem}\label{teoSumAux}
Let  $A=[(a_1,_1),\ldots, (a_k,m_k)]$ be a partition, and $b$, $p$ ($p \leq k$) two positive integers. $\!\!$
The application of the function $\!\!\textsc{ SumCompAux}$ to $A$, $b$ and $p$
always terminates and returns all the partitions  $A_1 \subseteq A$, \ldots, $A_s  \subseteq A$,  s.t.:
\begin{itemize}[nosep]
 \item the position of the elements in $A_1$, \ldots, $A_s$ is greater than (or equal to) $p$;
 \item the sum of the elements in each subpartition is $b$;
 \item no other subpartitions of $A$ with elements of position greater than (or equal to) $p$
  can be identified whose sum is $b$.
\end{itemize}
\end{theorem}
\begin{proof}
 Let $ k = \cardd{A} $. We have the following two cases.
 \begin{description}[nosep,leftmargin=0.4cm,align=left,labelwidth=1cm]
  \item[\small Case 1.] If $p>k$, the condition at line 1
   guarantees that no solution is provided: an empty set is returned and the function terminates.
  \item[\small Case 2.] If $p\leq k$, we proceed by induction on $q = k - p$.
   \begin{description}[nosep,leftmargin=0.4cm,align=left,labelwidth=1cm]
    \item[\small Case 2.1.] If $q=0$, then $p=k$.
     In this case, the only element of $A$ to be evaluated is $(a_k,m_k)$.
     The only possible solution is
     \begin{equation}\label{eq:IntDiv}
      v=b/a_k
     \end{equation}
     with $v \leq m_k$.
     We have the following subcases.
     \begin{description}[nosep,leftmargin=0.4cm,align=left,labelwidth=1cm]
      \item[\small Case 2.1.1.] If $b< a_k$, then Equation (\ref{eq:IntDiv}) has no solution
       and the function returns an empty set and terminates (line~2).
      \item[\small Case 2.1.2.] If $b\geq a_k$ and Equation (\ref{eq:IntDiv})
       has no integer solution or the integer solution is greater than $m_p$,
       then the function returns the empty set.
       In this case, in fact, no element ($a_p,t$) with $t \leq m_p$ exists s.t. $b=a_p t$.
       The first statement of the {\em else} branch (line 9) is executed
       and the function $\textsc{SumCompAux}$ is called passing a position $p+1>k$.
       For the part proved in Case 1, the empty set is returned back;
       therefore the condition at line 10 is false, the empty set is returned and the function terminates.
      \item[\small Case 2.1.3.] If $b\geq a_k$ and Equation (\ref{eq:IntDiv}) has an integer solution
       and the integer solution is lesser or equal than $m_p$,
       then the algorithm provides the set $\{(a_k,t)\}$ as result
       which solves exactly Equation~\ref{eq:IntDiv} with $t \leq m_k$ ($k=p$).
       In fact, in this case, the calculation of $quotient$ gives exactly $t$.
       In the {\em for} loop (line 5) there is only one case that satisfies the condition
       \begin{equation}\label{eq:IfStmt}
        i=quotient \qquad\text{and}\qquad quotient \cdot a_p  =b
       \end{equation}
       In this case, when the condition at line 6 is satisfied, the set $\{(a_k,quotient)\}$ is returned. Otherwise, the {\em else} branch (line 8) is executed and, using the same considerations of Case 2.1.2, $Result$ is left unaltered. At the end, the function correctly returns back the set $\{(a_k,quotient)\}$ and terminates.
       This closes the proof of Case 2.1
     \end{description}
    \item[\small Case 2.2.] If the thesis is true for $q-1$,
     then we have to prove that it holds also for $q$.
     \begin{description}[nosep,leftmargin=0.4cm,align=left,labelwidth=1cm]
      \item[\small Case 2.2.1.] If $b< a_{p}$, then Equation (\ref{eq:IntDiv}) has no solution
       neither for position $p$ nor for all positions greater than $p$:
       the empty set is returned and the function terminates.
      \item[\small Case 2.2.2.] If $b \geq a_{p}$, then we have the solution
       $ A'_j = [ (a_{p},m'_{p,j}),\ldots,(a_k,m'_{k,j}) ] $,
       where the elements with multiplicity $m'_{\_,j}=0$ have been discarded and
       \begin{equation}\label{eq:EqSum0}
        m'_{p,j} \leq m_{p} \qquad\text{and}\qquad b=\sum_{i=p}^{k} m'_{i,j} a_i  
       \end{equation}
       However, the last equation can be rewritten by extracting the part of position $p$ from the sum and moving it to the left member of the equation as follows: 
       \begin{equation}\label{eq:EqSum1}
        b-m'_{p,j}a_{p}  = \sum_{i=p+1}^{k} m'_{i,j} a_i\, .
       \end{equation}
       The second member of Equation (\ref{eq:EqSum1})
       is the sum of the components of the set $A''_j=\{(a_{p+1},m'_{p+1,j}),\ldots, (a_k,m'_{k,j})\}$.
       The first member of (\ref{eq:EqSum1}) is formed by varying the $m'_{p,j}$ from 0
       (when  $(a_{p},m'_{p,j})$ is not present) to a value $v$, where
       \begin{equation}\label{eqConstr1}
        v \leq b/a_{p} \qquad\text{and}\qquad v\leq m_{p}\, .
       \end{equation}
       The algorithm, after checking the condition at line 2, which is not satisfied,
       continues with the calculation of $quotient$.
       $quotient$ is exactly the maximum value as described in (\ref{eqConstr1}).
       The {\em for} loop (from line 5 to line 18) spans the $i$ values from 0 to $quotient$.

       When $i<quotient$ or $i=quotient$, but $quotient$ is not an exact divisor of $b$,
       the function enters in the {\em else} branch at line 8.
       In this branch, the set $\{A_1,\ldots,A_h\}$ is calculated by function $\textsc{SumCompAux}$
       with parameters $b - i  a_p$, $A$ and $p+1$ as showed in Equation (\ref{eq:EqSum1}).

       Due to the induction hypothesis, with this call,
       we have returned the set of partitions for position $p+1$ and $b'=b-i  a_{p}$,
       as stated in (\ref{eq:EqSum1})
       (the $m'_{p,j}=i$ and the sum of the elements of the returned set is $b'=b-i  a_{p}$).
       If this set is empty for the position $p+1$,
       it means that no solution exists also for position $p$.
       If the returned set is not empty, then it is added to $Result$ based on the value of $i$:
       \begin{itemize}[nosep]
        \item if $i\!=\!0$ by adding the returned set to $\!Result$ as it is because the left member of (\ref{eq:EqSum1}) is $\!b$,
        \item if $i>0$ we add to the returned set of position $p+1$, also the pair $(a_{p},i)$ because of  (\ref{eq:EqSum1}); 
       \end{itemize}
       Then, we can guarantee that $Result$ contains only and all the correct subpartitions
       for the case $(a_{p},i)$ with $i<quotient$ or $i=quotient$ but $quotient$ is not an exact divisor of $b$.
       When $i=quotient$ and $quotient$ is an exact divisor of $b$,
       then the only possible solution for this case is the pair $(a_{p},quotient)$ ($i=quotient$)
       because $b= quotient \cdot a_{p} $;
       then we add to $Result$ the pair $(a_{p},quotient)$.
       When the for loop ends, the function returns the $Result$ so calculated and terminates the execution.
       Thus, Case 2.2 is proved.
     \end{description}
   \end{description}
 \end{description}
 
 \vspace*{-.5cm}
\end{proof}

\begin{algorithm}[t]
\caption{$\textsc{SumComp}$}\label{alg1}
\begin{footnotesize}
\begin{algorithmic}[1]
	\Statex
	\Statex \textbf{Input:} Two partitions:  $A=\{(a_1,m_1),\ldots, (a_k,m_k)\}$, $B=[b_1,\ldots, b_m]$
	\Statex \textbf{Output:} The set of all $AB$-de\-com\-pos\-itions enabling the sum composition from $A$ to $B$  (when the sum composition exists, the empty set otherwise)
	
	\Statex
\IIf {$\sigma(A)  \not = \sigma(B)$}
 	   		 {\bf return} $\emptyset$
 	   \EndIIf \com{Item 5 of Lemma~\ref{lemma:1}}
 \IIf {$m > \card{A}$}
 	   		 {\bf return} $\emptyset$
 	   \EndIIf \com{Item 1 of Lemma~\ref{lemma:1}}
 	   \IIf {$a_1 > b_1$ {\bf or} $a_k > b_m$}
 	   		 {\bf return} $\emptyset$
 	   \EndIIf \com{Item 3 of Lemma~\ref{lemma:1}}
	
\IIf {$C = A \cap B$ s.t. $\card{C}>0$ and $\card{A} \leq 2m - \card{C} -1$}
 	   		 {\bf return} $\emptyset$
 	 \EndIIf \com{Theorem~\ref{theo:3}}
	\IIf {$\exists b \in B: \sum_{(a_i,m_i) \in A \land a_i \leq b}m_i a_i  < \sum_{b_j \in B \land b_j \leq b}b_j$}
		{\bf return} $\emptyset$
	\EndIIf \com{Theorem~\ref{theo:aGb}}
	\State $Result \leftarrow \emptyset$
         			\If { $m>1$}
				\State ${\mathcal A} \leftarrow \textsc{SumCompAux}(b_{1}, A, 1)$
				\IIf {$\mathcal A=\emptyset$}
					{\bf return} $\emptyset$
				\EndIIf
				\State $B'  \leftarrow  B \setminus [b_1]$
              			\For {{\bf each} $A_h \in {\mathcal A}$}
						\State ${\mathcal F} \leftarrow \textsc{SumComp}(A \setminus A_h, B')$
						\For {{\bf each} $f \in {\mathcal F}$}
							\State $Result \leftarrow Result \cup \{f \cup \{(b_{1},A_h)\}\}$
						\EndFor
       				\EndFor
			\Else
          				\State $Result \leftarrow \{ \{(b_{1},A)\}\}$
	   		\EndIf

\State {\bf return} $Result$
\Statex
\end{algorithmic}
\end{footnotesize}
\end{algorithm}

\subsection{The  $\textsc{SumComp}$ Algorithm}

The function $\textsc{SumCompAux}$ presented in the previous section
is the core function of our approach because  it generates all possible subpartition $A_1$, \ldots, $A_h$ of $A$
whose sum is equal to a $b$.
If $b$ is the first element of the partition $B$,
then each $A_i$ ($1 \leq i \leq h$) is the first element of an $AB$-de\-com\-pos\-ition.
Moreover, the problem reduces to find all $A'B'$-de\-com\-pos\-itions,
where $ A' = A \setminus A_i $ ($1 \leq i \leq h$) and $ B' = B \setminus [b_1] $.

Starting from this remark,
we have developed Algorithm~\ref{alg1}  ($\textsc{SumComp}$) for the generation of all $AB$-de\-com\-pos\-itions,
that works taking into account, in succession, the single elements of $B$.
In particular, this algorithm takes advantage of the theoretical results
obtained in Section~\ref{sec:properties}.

Algorithm~\ref{alg1} works as follows.
First, until line 6, it applies Items 5, 1 and 3 of Lemma~\ref{lemma:1},
and Theorems~\ref{theo:3} and \ref{theo:aGb} for checking necessary conditions of existence of sum composition.
When these conditions are met, the algorithm proceeds in the examination of all the elements in $B$,
otherwise, we can guarantee that no enabling de\-com\-pos\-itions can be determined.
Then, recursively, the algorithm  analyzes the length of  $B$.
When $\card{B}=m=1$, all elements of the original $B$ have been analyzed
and the pair $(b_1, A)$ is included in $Result$ (see line 18 of the branch {\em else}).
Indeed, since the condition at line $2$ is false, it means that $b_1$ is the sum of the elements of $A$.
When $\card{B}=m> 1$,
the function $\textsc{sumCompAux}$ is invoked (line 8) on the first element $b_{1}$ of $B$,
the current partition $A$ and the position $1$.
In this way, we determine the subpartitions of $A$ whose sum is $b_1$.
If the result of this invocation is the empty set,
it means that no sub-partition $A$ exist whose sum is  $b_1$ and then the empty set is returned.
Otherwise, we have identified all possible subpartition $\mathcal{A}$ for element $b_1$,
and $\textsc{sumComp}$ is recursively invoked for each of the partition $A_h \in \mathcal{A}$
on $A \setminus A_h$ and $B \setminus [b_1]$.
The recursive call will return the de\-com\-pos\-itions $\mathcal{F}$
for the elements of $B \setminus [b_1]$ starting from a partition $A$
from which $A_h$ have been removed.
At this point, all the elements in $\mathcal{A}$ should be combined
with all the elements in $\mathcal{F}$ in order to determine
all the $AB$-de\-com\-pos\-itions between $A$ and $B$ (line 14 of the algorithm).

\begin{figure}[t]
\centering
\includegraphics[scale=0.68]{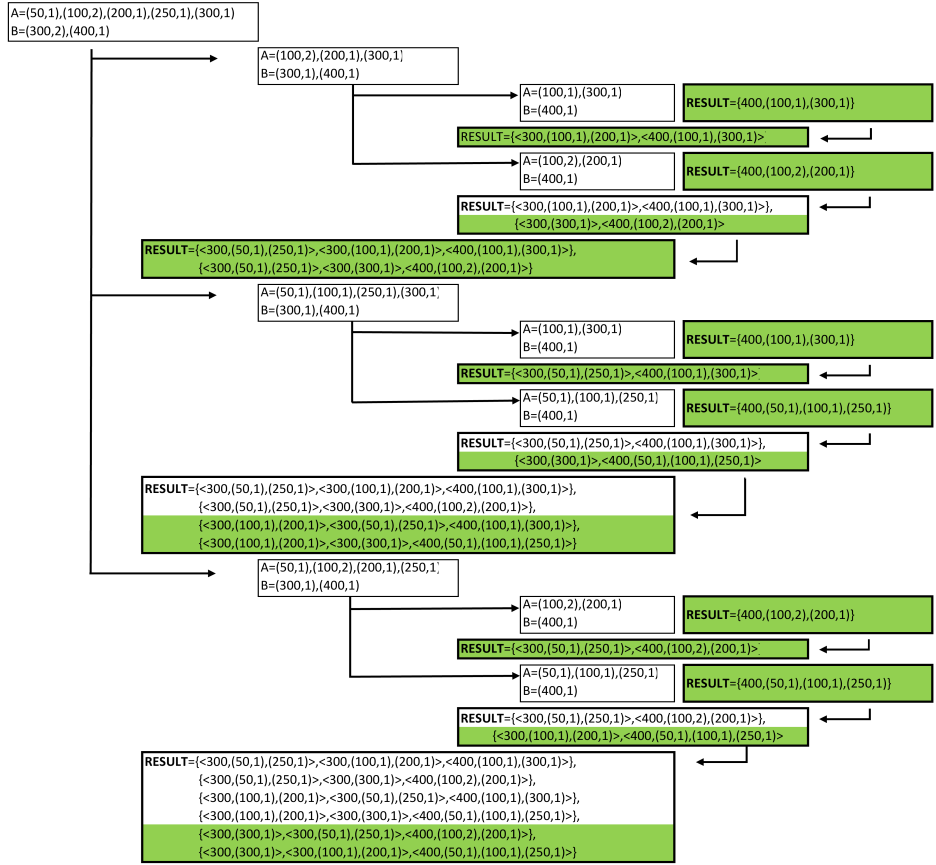}
\caption{Recursive calls of the $\textsc{sumComp}$ Algorithm}\label{fig:sumComp}
\end{figure}

As shown by the following example, the application of  $\textsc{sumComp}$ returns the set $Result$ that contains all the $AB$-de\-com\-pos\-itions that can be determined from the initial $A$ and $B$.

\begin{example}\label{ex:sumcomp}
 Let 
 $A=[(50,1),(100,2),(200,1),(250,1), (300,1)]$ and $B = [300, 300, 400]$ the partition of Example~\ref{ex:es2}.
 The first line of Figure~\ref{fig:sumComp} reports the initial partitions $A$ and $B$.
 The first vertical arrow on the left leads to the three possible $A_1$, $A_2$, $A_3$ subpartitions
 whose sum is $300$ (first element of $B$) that can be obtained from $A$
 by calling $\textsc{SumCompAux}$ at line 8 of the algorithm.
 Starting from them, the input of recursive calls to $\textsc{SumComp}$ (line 12)
 are determined as follows $(A \setminus A_1, B \setminus [300])$,
 $(A \setminus A_2, B \setminus [300])$, and $(A \setminus A_3, B \setminus [300])$.
 Then, the second vertical line represents the subpartitions whose sum is $300$
 (second element of the initial $B$) that can be obtained from the current $A$.
 At this point, the backward arrow leads to the partial and final results (in total six)
 obtained by identifying the subpartitions of $A$ whose sum is $400$.
 These partial results are accumulated in the $Result$ variable from the recursive calls.
 The results in green are the new one that are added to the existing ones (in white) in the current step.
 \ebox
\end{example}

\begin{theorem}\label{theo:correct1}
 Let $A=[(a_1,m_1),\ldots, (a_k,m_k)]$  and $B=[b_1,\ldots, b_m]$ be two partitions.
 The application of Algorithm~\ref{alg1} to $A$ and $B$
 always terminates and returns all possible $AB$-de\-com\-pos\-itions
 (the empty set when $ A \not\leq B $).
\end{theorem}
\begin{proof}
Each result $\mathcal S_i$ of this algorithm has the the form $\mathcal S_i=\{(b_1,A_{i,1}),\ldots,$ $(b_m,A_{i,m})\}$ and 
 is  an $AB$-de\-com\-pos\-ition. To prove this claim  we have to show that:
\begin{enumerate}[nosep]
 \item  $S_i$ is an $AB$-de\-com\-pos\-ition when $A \leq B$, the empty set otherwise.
 \item $Result$ contains only $AB$-de\-com\-pos\-itions.
 \item $Result$ contains all $AB$-de\-com\-pos\-itions.
\end{enumerate}
The algorithm termination is guaranteed by the verification of these three items.
\begin{description}[nosep,leftmargin=0.4cm,align=left,labelwidth=1cm]
 \item[\small Item 1.] The algorithm first checks if the sufficient conditions of non existence of sum composition
  (lines 1 to 5) are satisfied.
  When one of these conditions is true, the empty set is returned and the algorithm terminates.
  Otherwise, the proof of item 1 is conducted by induction on $m=\card{B}$
  \begin{itemize}[nosep]
   \item If $m=1$, then the solution is exactly the partition $A$.
    In fact, since $m=1$ and the condition of line $1$ is false,
    we have that $\sum_{i=1}^{k} m_i  a_i = \sum_{j=1}^{m}b_j = b_1$.
    $Result$ is initialized with the empty set (line 6)
    and, since the condition of the {\em if} at line 8 is not satisfied,
    then the instruction at line 18 is executed
    and assigns to $Result$ the single solution $\{(b_1,A)\}$.
  \item We assume the property true for $\card{B}=m-1$ and we prove it for $\card{B}=m$.
   $Result$ is initialized with empty set.
   The {\em if} at line 7 is satisfied,
   then the invocation of the $\textsc{SumCompAux}$ function
   returns the set ${\mathcal A}$ of all the subpartitions of $A$ whose sum is $b_1$,
   if they exist (see Theorem~\ref{teoSumAux}).
   If ${\mathcal A}$ is the empty set (condition at line 9),
   it means that no combination of elements of $A$ can give as sum $b_1$
   and then $B$ is not sum composition of $A$ and the empty set is returned.
   Instead, if ${\mathcal A \neq \emptyset}$,
   then the function spans on all the $A_h$ that are the outputs of $\textsc{SumCompAux}$
   (from line 11 to line 16).
   This is to guarantee that each combination of possible solutions is taken into consideration.
   For each $A_h$, the algorithm recursively calculates all the enabling decompositions
   for $A \setminus A_h$ and $B \setminus [b_1]$.
   The two inputs are not empty.
   In fact, $B \setminus [b_1]$ has at least the element $b_2$ (being $m>1$).
   Moreover, $A \setminus A_h \ne \emptyset$,
   being $ \sigma(A_h) = b_1 < \sigma(B) = \sigma(A) $,
   i.e. $ \sigma(A_h) < \sigma(A) $.
   So, by the inductive hypothesis, the set ${\mathcal F}$ contains
   all $A \setminus A_h, B'$-de\-com\-pos\-itions.
   The $Result$ is obtained (line 14) by adding, for each solution $f \in {\mathcal F}$,
   the pair $(b_1,A_h)$.
   The obtained result is thus an $AB$-de\-com\-pos\-ition.
  \end{itemize}
 \item[\small Item 2.]
  Consider a generic solution in $Result$ and we show that it is valid.
  The generic solution $s \in Result$, by the induction hypothesis,
  has the form $ s=f \cup \{(b_1,A_1)\} $.
  However, since $f$ is one of the output of Algorithm~\ref{alg1} ($A \setminus A_1, B \setminus [b_1]$),
  it can be written as $ f =\{(b_2,A_2),\ldots,(b_m,A_m)\} $.
  Then $ S=\{(b_1,A_1)\} \cup \{(b_2,A_2),\ldots,(b_m,A_m)\} $
  with the property $ \sigma(A_1) = b_1 $ and
  $ \sigma(A_i) = b_i $ for every $ i = 2, \ldots, m $.
  So $ A \leq B $.
 \item[\small Item 3.]
  By absurd, suppose that a solution $ S=\{(b_1,A_1),\ldots,(b_m,A_m)\} $ is not in $Result$.
  Because at line 8 we have in ${\mathcal A}$ all the pairs ($b_1, A'_h$)
  where $\sigma(A'_h) = b_1 $ (see Theorem~\ref{teoSumAux}),
  then the pair ($b_1, A_1$) was provided by the function $\textsc{SumCompAux}$.
  Then, the recursive call of Algorithm~\ref{alg1} with the parameters $A \setminus A'_1$ and $B \setminus [b_1]$,
  by inductive hypothesis, returns all $(A \setminus A'_1,B \setminus [b_1])$-de\-com\-pos\-itions.
  So, it should also contain $\{(b_2,A_2), \ldots, (b_m, A_m)\}$.
  The union of these two sets is thus an $AB$-de\-com\-pos\-ition,
  in contradiction with the initial hypothesis.
\end{description}
 This concludes the proof of the theorem.
\end{proof}

\section{The  Existential Sum Composition Algorithm}\label{sec:existence}

In the previous section, Algorithm~\ref{alg1}
was developed for generating all $AB$-de\-com\-pos\-itions
of two given integer partitions $A$ and $B$.
However, for several problems (as those mentioned in the introduction)
we are not interested on all of them,
but only on the fact that $B$ is sum composition of $A$,
that is, if there exists at least one $AB$-de\-com\-pos\-ition.
In this section, we develop the Algorithm $\textsc{SumCompExist}$ to provide an answer to such a problem.
This algorithm relies on the function $\textsc{SumCompExistAux}$,
which checks the sufficient conditions for having $A \leq B$.
This function, however, is invoked
only after applying some \emph{simplifications},
according to the theorems of Section~\ref{subsection:existence}
for removing elements from $A$ and $B$ while preserving the relation of sum composition.

Algorithm $\textsc{SumCompExist}$ checks in cascade several properties
by applying items 1, 3, and 5 of Lemma~\ref{lemma:1}.
Then, it applies the first \emph{simplification}
by eliminating the equal values from the partition $A$ and $B$ (Theorem~\ref{cor1}).
Then, on the returned sets (if not empty)
it applies Theorem~\ref{theo:3} and then Theorem~\ref{theo:simpl2}.
Finally, it performs the check related to Theorem~\ref{theo:aGb}.
After these checks and \emph{simplifications}, the function $\textsc{SumCompExistAux}$ is invoked.

The function $\textsc{SumCompExistAux}$ starts by checking the condition of Item 5 of Lemma~\ref{lemma:1} (line 12).
Whenever the condition is verified, the $false$ value is returned
because no $AB$-de\-com\-pos\-itions can be obtained.
Otherwise, the length of $B$ is checked.
If $\card{B}=1$ (line 14), then we are sure that an $AB$-de\-com\-pos\-ition has been identified.
Otherwise, when $\card{B}>1$, the function (line 14)
determines the partitions $[A_1,\ldots,A_h]$ for $b_1$ by invoking
the function $\textsc{SumCompAux}$ (described in previous section).
Whenever $h \not = 0$, we have $\sigma(A_i)=b_1$ for every $ i = 1, \ldots, h $,
and, for each $A_i$, $\textsc{SumCompExistAux}$ is recursively invoked
on $A \setminus A_i$ and $B \setminus [b_1]$
for checking the relation of sum composition between these two partitions (line 17).
When one of these recursive function returns $true$,
it means that $A \leq B$ and then the $true$ value is returned (line 18).
Whenever none of them returns $true$,
it means that no $AB$-de\-com\-pos\-itions exist and the $false$ value is returned (line 20).

\begin{algorithm}[t]
\caption{$\textsc{SumCompExist}$}\label{alg2}
\begin{footnotesize}
\begin{algorithmic}[1]
	\Statex
	\Statex \textbf{Input:} Two partitions:  $A=\{(a_1,m_1),\ldots, (a_k,m_k)\}$, $B=[b_1,\ldots, b_m]$
	\Statex \textbf{Output:} $P=true$ (when the sum composition exists)
\Statex  \hspace*{1.3cm}	 $P=false$ otherwise ($B$ is not sum composition of $A$)
	
	\Statex
	 \IIf {$\sigma(A)   \not = \sigma(B)$}
 	   		 {\bf return} $false$
 	   \EndIIf \com{Item 5 of Lemma~\ref{lemma:1}}
	 \IIf {$m > \card{A}$}	 
 	   		 {\bf return} $false$
 	   \EndIIf \com{Item 1 of Lemma~\ref{lemma:1}}
 	   \IIf {$a_1 > b_1$ {\bf or} $a_k > b_m$}
 	   		 {\bf return} $false$
 	   \EndIIf \com{Item 3 of Lemma~\ref{lemma:1}}
\State $(A, B) \leftarrow (A \setminus C, B \setminus C)$ \com {Where, $C= A \cap B$ - partition \emph{simplification} - Corollary~\ref{cor1}}
	 \IIf {$A = \emptyset$}
		{\bf return} $true$
	\EndIIf \com{partitions A and B are equal}
	 \IIf {$\card{A} \leq 2  \card{B} -1$}
 	   		 {\bf return} $false$
 	 \EndIIf \com{Theorem~\ref{theo:3}}
	\IIf {$a_{k}\!>\!b_{m-1}$}
		{$(A, B) \!\leftarrow\! (A \setminus [a_{k}], B \setminus [b_{m}]\cup [b_{m}-a_{k}])$}
	\EndIIf \com {Theorem~\ref{theo:simpl2}}
	\IIf {$\!\exists h\!\in\!A\!\cup\!B\!:\!\sum_{(a_i,m_i) \in A \land a_i \leq h}m_i\! \!a_i\!<\! \sum_{b_j\!\in\!B \land b_j \leq h}b_j$}
		{\bf return} $false$
	\EndIIf \com{Theorem~\ref{theo:aGb}}

	\State $Result \leftarrow \textsc{SumCompExistAux}(A, B)$
	\State \Return $Result$

\Statex \Function {SumCompExistAux}{$A, B$}
	 \IIf {$\sigma(A)   \not = \sigma(B)$}
		{\bf return} $false$
	\EndIIf
		\IIf {$\card{A}=1$}  {\bf return} $true$
	\EndIIf
	\State $\{A_1,\ldots,A_h\}\leftarrow \textsc{SumCompAux}(b_1, A, 1)$
	 \State $B'=B \setminus [b_1]$
           \For {{\bf each} $A_i \in \{A_1,\ldots,A_h\}$}
		\State $Result \leftarrow \textsc{SumCompExistAux}(A \setminus A_i, B')$
		\IIf {$Result = true$}
			{\bf return} $true$
		\EndIIf
	\EndFor
\State {\bf return} $false$
\EndFunction
\end{algorithmic}
\end{footnotesize}
\end{algorithm}

\begin{example}
Let 
 $A=\{(50,1),(100,2),(200,1),(250,1),(300,1)\}$ and $B = [300, 300, 400]$ be the partitions of Example~\ref{ex:es2}.
 Since $A \cap B = [300]$, by applying the code at line 4 of the algorithm,
 we obtain $A=\{(50,1),(100,2),(200,1),(250,1)\}$ and $B = [300, 400]$.
 By considering the first $300$ in $B$,
 the subpartitions of $A$ whose sum is $300$ are only $\{(50,1),(250,1)\}$ and $\{(100,1),(200,1)\}$.
 Since the invocation of the recursive function on $A=\{(100,2),(200,1)\}$ and $B = [400]$ returns true,
 the algorithm returns the existence of a de\-com\-pos\-ition between the two subpartitions.
 Therefore, the process in this case is  faster than the one described in Example~\ref{ex:sumcomp}. \ebox
\end{example}

The following theorem provides the correctness of the algorithm.
\begin{theorem}\label{theo:Existence}
 The application of Algorithm $\textsc{SumCompExist}$ to the partitions $A$ and $B$ always terminates
 and it returns $true$ whenever $B$ is sum composition of $A$.
\end{theorem}

\begin{proof}
 The function $\textsc{SumCompExist}$ can be divided in two parts:
 the first one (from line 1 to line 9) verifies the sufficient conditions
 for which $B$ is not sum composition of $A$.
 The only exceptions are lines 4 and 7 in which the two partitions are simplified.
 Whenever, after removing from $A$ and $B$ their intersection,
 the obtained partition $A$ is empty,
 it means that $A=B$ and trivially $B$ is sum composition of $A$ (line 5).

 The second part of $\textsc{SumCompExist}$ is simply
 returning the result from $\textsc{SumCompExistAux}$
 that can be proved by induction on the length of $B$.
 \begin{description}[nosep,leftmargin=0.4cm,align=left,labelwidth=1cm]
  \item[\small Case 1.] If $\card{B}=1$, then we have two possibilities.
   If the condition at line 12 is verified,
   then the function returns $false$ as it has to be by Item 5 of Lemma~\ref{lemma:1}.
   On the contrary, the condition on line 13 is verified
   (because $b_1=\sigma(A)$) and $true$ is returned.
   In both cases the algorithm terminates the execution.
  \item[\small Case 2.] If $\card{B}>1$, we assume true for a $m=\card{B}-1$,
   and we  prove for $m+1$.
   Also in this case, if the condition at line 12 is verified,
   then the function returns $false$ and terminates (according to  Item 5 of Lemma~\ref{lemma:1}).
   The condition at line 13 is false and thus function $\textsc{SumCompAux}$ is invoked on $b_1$, $A$, and $1$.
   It returns the subpartitions $\{A_1,\ldots,A_h\}$ of $A$ s.t.
   $b_1=\sigma(A_j)$ for every $ j = 1,\ldots,h $ (Theorem~\ref{teoSumAux}).
   If this set is empty, it means that $b_1$ cannot be associated with any element of $A$
   and thus $B$ is not sum composition of $A$.
   Therefore, the function jumps the {\em for} loop (lines 16--19),
   returns $false$ and terminates.
   On the contrary, if this set is not empty,
   then the {\em for} loop (lines 16--19) is entered
   and each $A_i \in \{A_1,\ldots,A_h\}$ is processed
   by recursively invoking $\textsc{SumCompExistAux}$ on $A \setminus A_i$ and $B \setminus [b_1]$.
   By inductive hypothesis, this function returns the existence of a sum composition between its parameters.
   When one of the returned value is $true$, the function also returns $true$.
   Indeed, $\sigma(A_i) = b_1$ and $B \setminus [b_1]$ is sum composition of $A \setminus A_i$
   then, by Lemma~\ref{lemma:sumnm1}, $B$ is sum composition of $A$.
   Whenever none of the recursive invocations of $\textsc{SumCompExistAux}$ returns $true$,
   it means that $B$ cannot be sum composition of $A$,
   and the $false$ value is returned and the algorithm terminates.
 \end{description}
 
 \vspace*{-.5cm}
\end{proof}

\section{Experimental Results}\label{sec:exp}

The algorithms described in the previous sections have been experimentally evaluated by considering synthetic data\-sets with specific characteristics for  comparing their behaviors in the cases in which we know the existence of sum composition and the cases in which no sum composition can be determined.
In the generation of the partition $A$, we have considered the possibility to randomly generate $n$ samples in specific ranges of values with or without repetitions. In the following we denote with $A^{n}_{[min,max]}$ the partition $A$ whose length is $n$ and $\forall a \in A, min \leq a \leq max$.
For the generation of a  partition $B$ with $m$ elements ($\!2\! \leq\! m\! <\! n$) we have considered  $\!A^{n}_{[min,max]}\!$ and applied one of the following rules.
\begin{enumerate}[nosep]
\item[R1] Randomly split the $A^{n}_{[min,max]}$ in $m$ non-empty partitions  and then determine the values of $B$ by summing up the integers in each partition.
\item[R2] Generate the sum of the values of $A$,
 randomly generate $m$ distinct ordered values as follows:
 $v_1,\ldots, v_{m-1}$ in the range $[1, \sigma(A)]$, and $v_m=\sigma(A)$.
 Assuming $v_0=0$, each $b_i$ is obtained as $b_i=v_i-v_{i-1}$.
\end{enumerate}
Using rule {\tt R1}, we are guaranteed of the existence of sum composition,
whereas by rule {\tt R2}, we have that $\sigma(A)=\sigma(B)$,
but the condition of sum composition is not always guaranteed.
Notice that in case of non-existence of a sum composition between the partitions $A$ and $B$,
the execution cost of the exhaustive and existential algorithms should be comparable
because the entire space of possible solutions has to be explored.
In this evaluation, we also wish to study the impact of the properties
presented in Section~\ref{sec:properties}
for reducing the number of recursive calls introduced in the algorithms.
All experiments have been conducted 100 times for each type of generation of $A$ and $B$ and the reported results  (both in terms of execution time and number of enabling de\-com\-pos\-itions) correspond to their average. The experiments have been executed on a desktop with medium capabilities (Intel Core i7-4790 CPU at 3.60 GHz, processor based on x64 bit, 8 GB RAM).

\subsection{Experiments with the $\textsc{SumComp}$  Algorithm}
In the evaluation of the performances of the $\textsc{SumComp}$ Algorithm we have considered different combinations of the parameters for the generation of $A^{n}_{[min,max]}$ and the rule $R1$ for the generation of $B$. Indeed, the worst case corresponds on the existence of the sum composition and the need of generating all possible $AB$-de\-com\-pos\-itions.

The first experiments have been devoted to the identification of the execution time and number of enabling de\-com\-pos\-itions by varying the length of $A$ and $B$.  Due to the explosion of different combinations, and the limitation of the machine used for the experiments, we were able to consider partitions of $A$  with $\card{A}\leq 23$. Indeed, with $\card{A}=24$, memory allocation issues were raised ($>4.8$ GB memory used).  Moreover, we have considered three ranges of values  ($[1,100]$, $[1,150]$, $[1,200]$) in which the values of $A$ have been randomly chosen.

Figure~\ref{fig:SumCompTime1-200} shows a 3-dimensional comparison of the execution times varying the length of $A^{n}_{[1,200]}$ (with $3 \leq n \leq 23$) and the length of $B$  (with $2 \leq \card{B} \leq 22$). The range $[1,200]$ has been chosen because it is the most representative among the considered ones, in fact the results in this range dominate the results in the other ones.
The graphic shows that the execution times tend to increase when  $\card{B}$ is much  smaller than $\card{A}$. This makes sense because a higher number of possible combinations of values in $A$ can have as sum the values in $B$. The execution times with these examples is affordable, but by increasing of a single unit the length of $A$, the memory is not sufficient for verifying all the different combinations and the results cannot be reported. 

\begin{figure}[t]
\centering
\subfigure[]{\label{fig:SumCompTime1-200}\includegraphics[scale=.262]{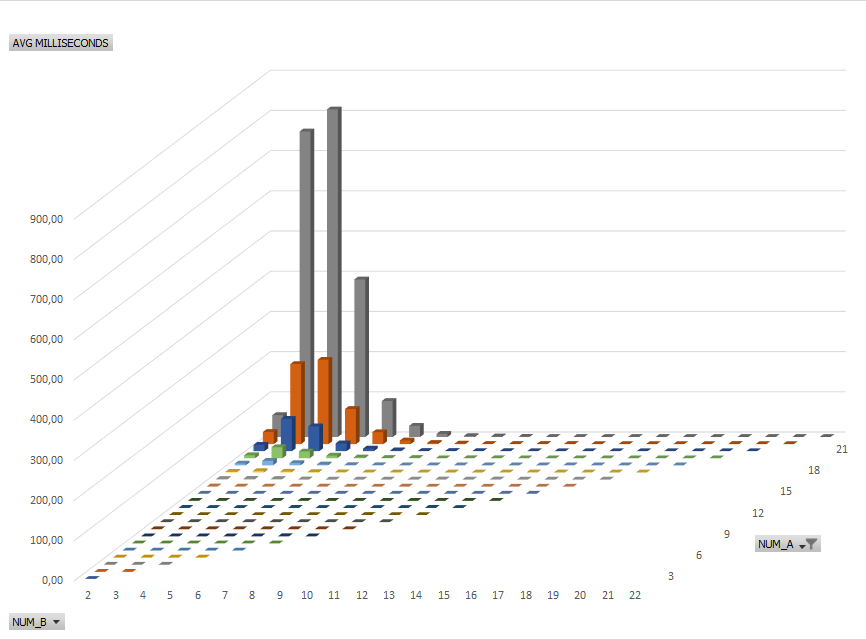}} \ \
\subfigure[]{\label{fig:SumCompNumberTimeB41-200}\includegraphics[scale=.262]{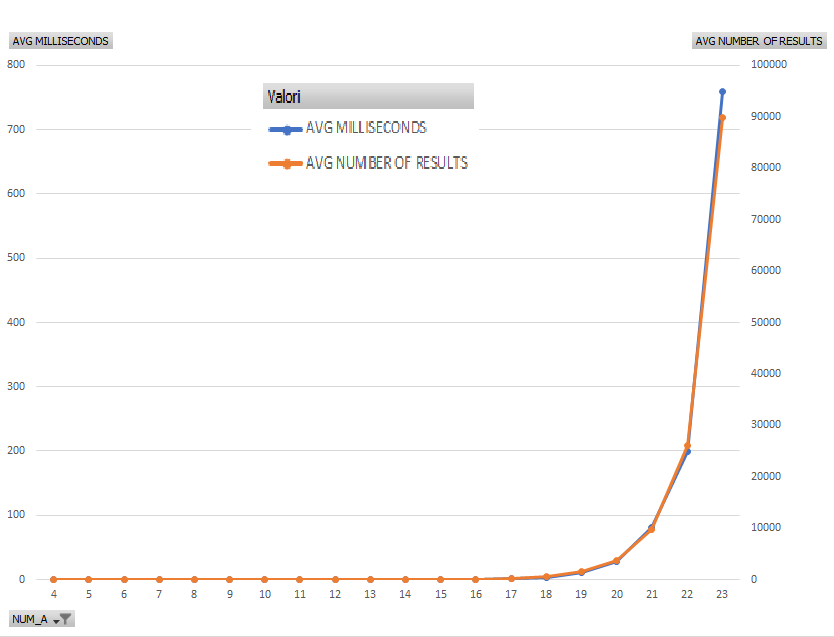}} \\
\subfigure[]{\label{fig:SumCompNumberTimeA231-200}\includegraphics[scale=.262]{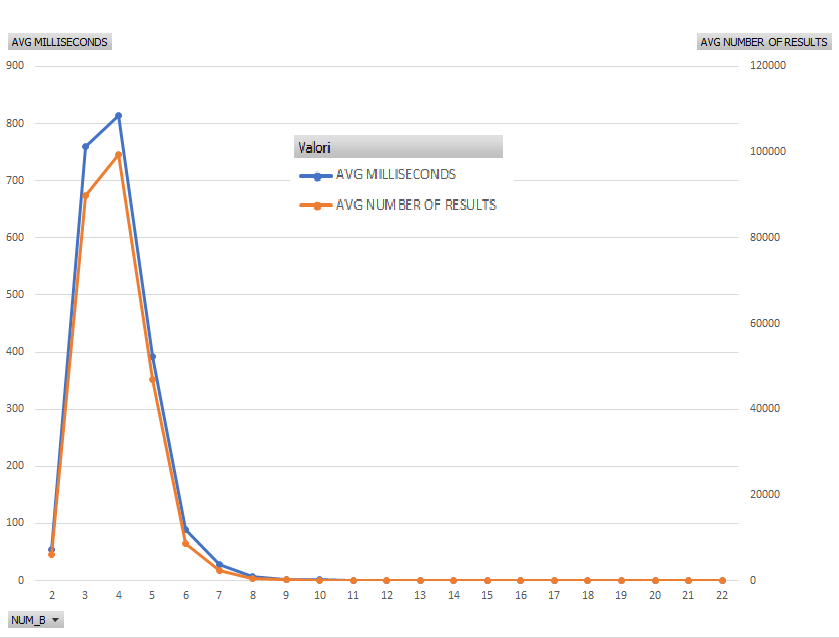}}\ \
\subfigure[]{\label{fig:SumCompNumberA23DistinctRange}\includegraphics[scale=.262]{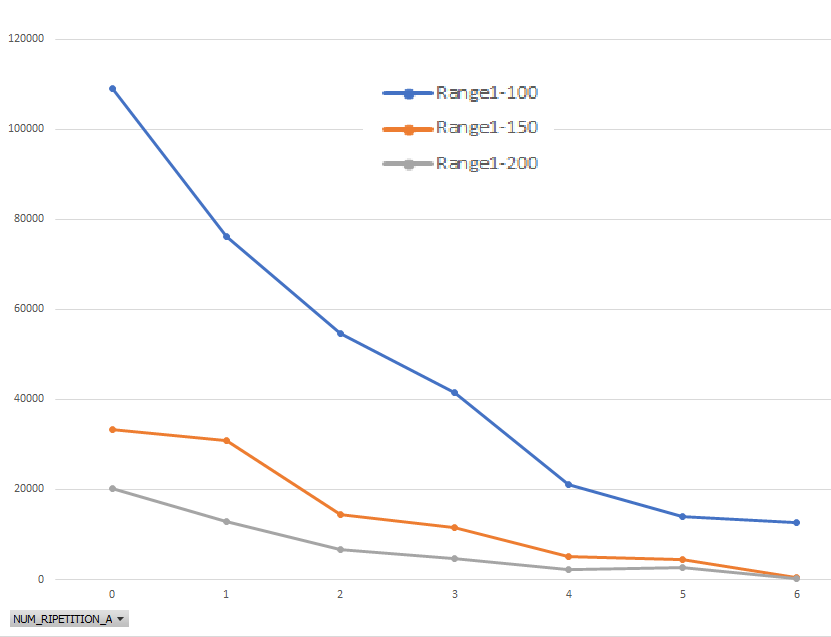}}
\caption{
Experiments with function $\textsc{sumComp}$ by applying rule $R1$.
  $(a)$ Execution times varying the length of $A^{n}_{1,200}$ and $B$  ($1 \leq \card{A}, \card{B} \leq 22$);
$(b)$ number of enabling de\-com\-pos\-itions and execution times varying the length of $A^{n}_{1,200}$ (with $1 \leq n \leq 22$) and fixing $\card{B}=4$;
$(c)$  number of enabling de\-com\-pos\-itions and execution times for $\card{A}=23$ and  varying the length of $B$  ($1 \leq \card{B} \leq 22$); and, $(d)$
execution times  with $\card{A}=23$ and $\card{B}=2$  varying the  number of repetitions in $A$ and the considered ranges ($[1,100]$, $[1,150]$, $[1,200]$).}\label{fig:SumCompTot}
\end{figure}

A slice of  Figure~\ref{fig:SumCompTime1-200}, by fixing  $\card{B}=4$, is reported in Figure~\ref{fig:SumCompNumberTimeB41-200} with also the number of identified enabling de\-com\-pos\-itions.  The graphic shows that the number of solutions (as well as the execution times) increase exponentially. Specifically, around $100.000$ enabling de\-com\-pos\-itions in average can be determined for $\card{A}=23$ in an elapsed time of 900 ms in average. Moreover, we can identify a correlation between the number of solutions and the execution times (correlation index $\approx 91\%$ on $\card{A}$).
Another slice of Figure~\ref{fig:SumCompTime1-200}, by fixing $\card{A}=23$, is reported in Figure~\ref{fig:SumCompNumberTimeA231-200}. In this case we can note that the maximum execution time (and also the number of enabling de\-com\-pos\-itions) is obtained for low values of $\card{B}$. Indeed, the few values contained in $B$ can be obtained by summing different combinations of values in $A$. When   $\card{B}$ increases the possible combinations deeply decrease.

If we now compare all the slices that can be obtained from Figure~\ref{fig:SumCompTime1-200} we can make the following observations. The execution time for the slice corresponding to $\card{A}=20$  is around $30$ ms while the one for the slice  $\card{A}=23$ moves to $900$ ms with a significant increase of time. Moreover, the highest execution time for the slice corresponding to $\card{A}=20$ is reached with $\card{B}=3$, whereas for the case $\card{A}=23$ the maximum execution time is reached with $\card{B}=4$.  Moreover, in the considered range of values we can state that the maximum number of enabling de\-com\-pos\-itions follows the law $5.5 \leq \frac{\card{A}}{\card{B}} \leq 7$. However, this observation holds only for the considered partitions $A$ and ranges. A deeper analysis is required for proving the general validity of this claim.  Table~\ref{tbl:enabling} reports the number of solutions (in average) by changing the range of values from which the integers in $A$ are chosen.  We can observe that the number of enabling de\-com\-pos\-itions is higher for values of $A$ chosen in the range $[1,100]$. This result is not intuitive and needs further investigations.

\begin{table}[b]
\begin{footnotesize}
\begin{center}
  \begin{tabular}{ r | r | r | r}
    \hline
    \# A & Range 1-100 & Range 1-150 & Range 1-200 \\ \hline
18 & 1,168 & 447 & 251 \\ \hline
19 & 4,086 & 1,520 & 667 \\ \hline
20 & 13,901 & 4,261 & 2,097 \\ \hline
21 & 38,352 & 18,815 & 7,810 \\ \hline
22 & 159,581 & 54,395 & 26,285 \\ \hline
23 & 477,449 & 198,690 & 99,552 \\    \hline
  \end{tabular}
\end{center}
\end{footnotesize}
\caption{Number of $AB$-de\-com\-pos\-itions by varying the length of $A$ and the range of values}\label{tbl:enabling}
\end{table}

Figure~\ref{fig:SumCompNumberA23DistinctRange}  shows the impact on the execution time of the presence of repeated elements in $A$. By increasing the number of  repetitions of values in $A$ (0 means no duplication, 1 one value duplicated and so on) we observe that the execution time  decreases for each range of values.
This means that our way to represent partitions  has positive impact on the performances.

\begin{figure}[t]
\centering
\subfigure[]{\label{fig:SumCompExistGlobal}\includegraphics[scale=.248]{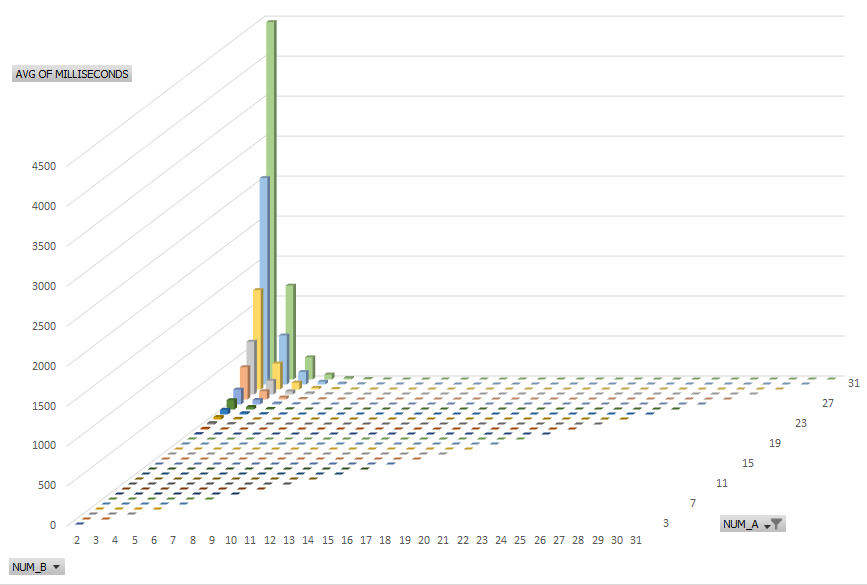}} \ \hspace*{6pt} \
\subfigure[]{\label{fig:SumCompExistGlobalSliceB2}\includegraphics[scale=.248]{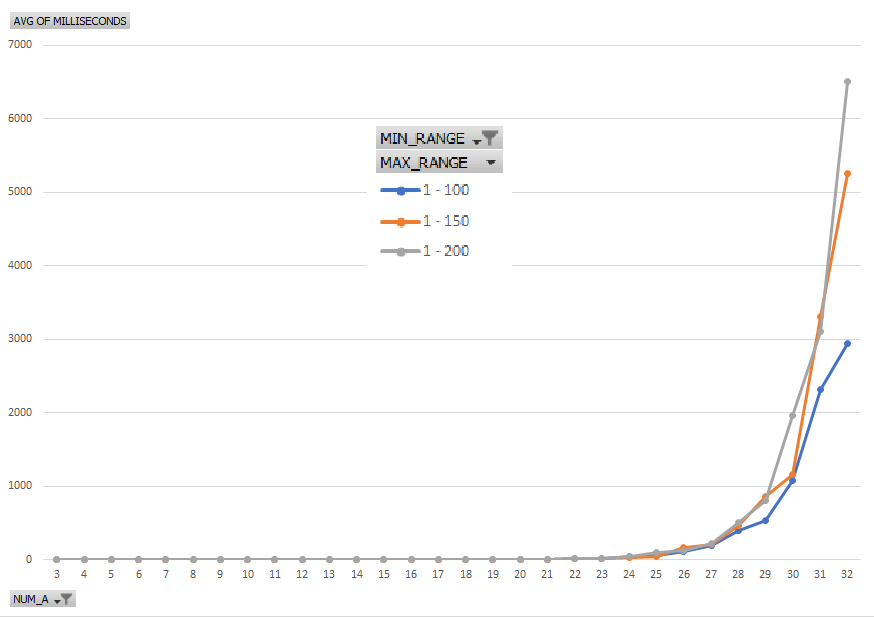}} \\
\subfigure[]{\label{fig:SumCompExistGlobalSliceA32}\includegraphics[scale=.248]{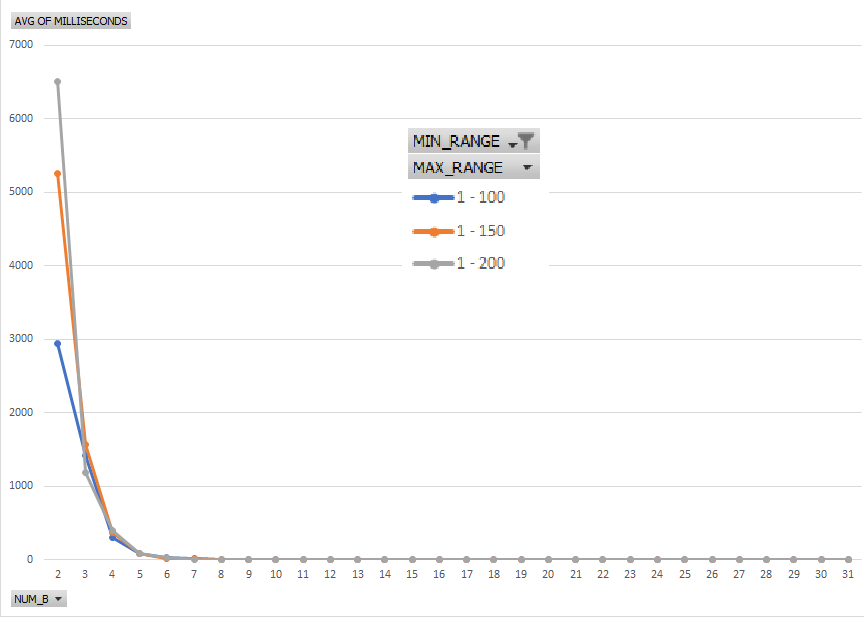}}\ \hspace*{6pt} \
\subfigure[]{\label{fig:SumCompExistA32DistinctRange}\includegraphics[scale=.248]{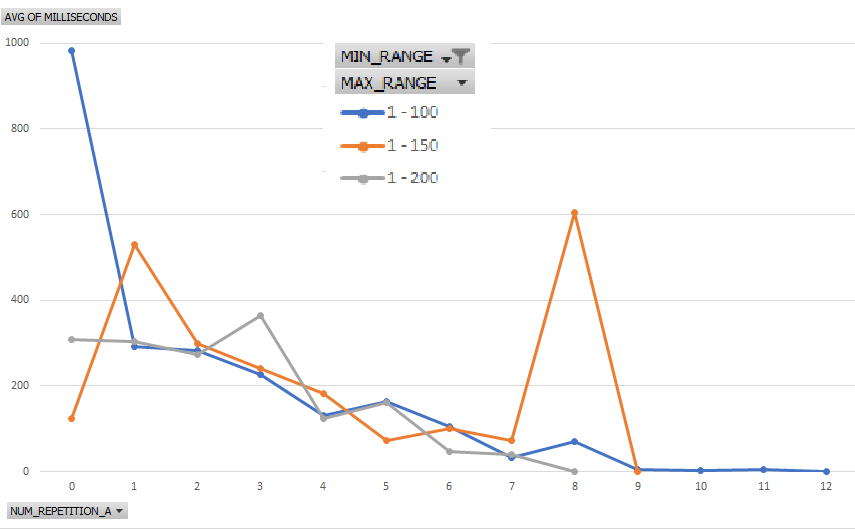}}
\caption{Experiments with function $\textsc{sumCompExist}$.  (a) Execution times varying the length of $A^{n}_{1,200}$ ($3 \leq n \leq 32$) and the length of $B$  ($2 \leq \card{B} \leq 31$);
(b) number of enabling de\-com\-pos\-itions and execution times varying the length of $A^{n}_{1,200}$ ($3 \leq n \leq 32$) and for $\card{B}=4$;
(c)  number of enabling de\-com\-pos\-itions and execution times for $\card{A}=32$ and  varying the length of $B$  ($2 \leq \card{B} \leq 31$);
execution times  with $\card{A}=32$ varying the  number of repetitions in $A$ and the considered ranges ($[1,100]$, $[1,150]$, $[1,200]$).}\label{fig:SumCompExist}
\end{figure}

\subsection{Experiments with the $\textsc{SumCompExist}$  Algorithm}
In the experiments with the existential algorithm we have considered the rule $R2$ for the generation of the partition $B$. Indeed, in this case we are interested in evaluating the behavior of the algorithm also when no solutions exist. Actually, the lack of an enabling de\-com\-pos\-ition corresponds to the worst case because the entire space of solutions should be explored before reaching to the conclusion. With this algorithm we have identified an issue of lack of memory too (when $\card{A}=33$). However, since this algorithm does not require to carry all the enabling de\-com\-pos\-itions, we were able to conduct experiments with a partition $A$ with maximal length of $32$.

%
%

\begin{figure}
  \centering
  \begin{minipage}[b]{0.45\textwidth}
    \includegraphics[scale=.3]{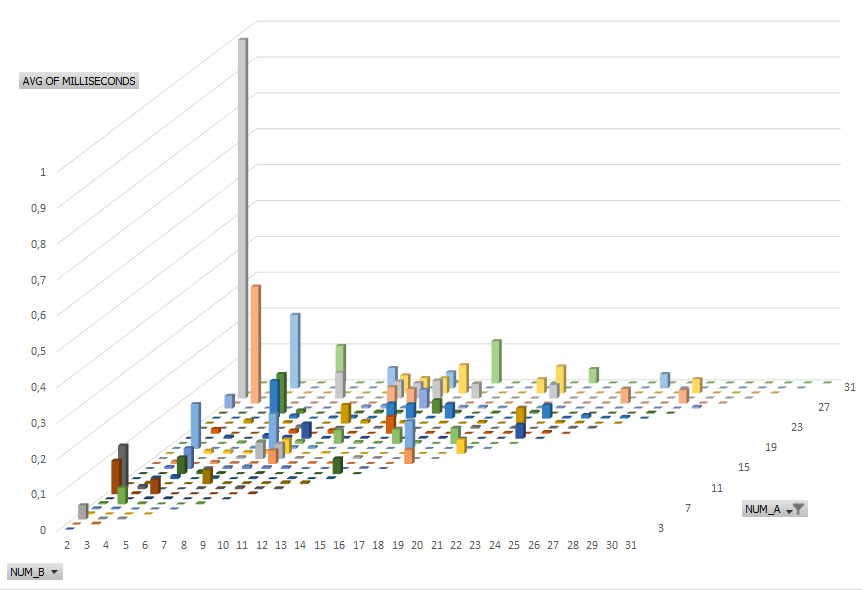}
\caption{Execution times of $\textsc{sumCompExist}$ when enabling de\-com\-pos\-itions do not exist -- varying the length of $A^{n}_{1,200}$  and the length of $B$  (with $3 \leq \card{A} \leq 32, 2 \leq \card{B} \leq 31$)} \label{fig:SumCompExistNotFound}
  \end{minipage}
  \ \hspace*{18pt}\
  \begin{minipage}[b]{0.45\textwidth}
   \includegraphics[scale=.33]{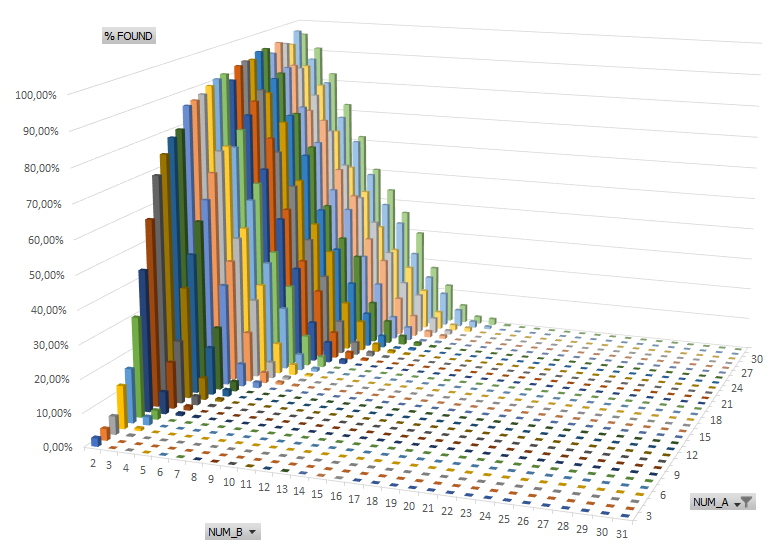}
\caption{Frequency of existence of sum composition w.r.t. the generated datasets by varying the  length of $A^{n}$ for all ranges and the length of $B$  (with $3 \leq \card{A} \leq 32, 2 \leq \card{B} \leq 31$)} \label{fig:SumCompExistBasePerFound}
  \end{minipage}
\end{figure}

Figure~\ref{fig:SumCompExistGlobal} shows a 3-dimensional comparison of the execution times varying the length of $A^{n}_{[1,200]}$ ($3 \leq n \leq 32$) and the length of $B$  ($2 \leq \card{B} \leq 31$). The skyline of the time execution of the existential algorithm is similar to the exhaustive version. However, the absolute times are considerably smaller. Indeed, the worse execution time for the exhaustive algorithm is around 900 ms for $\card{A}=23$, whereas for the existential algorithm for a partition $A$  with the same length is around $17$ ms. In case of existence, moreover, we are able to handle a partition $A$ with 32 elements. Thus, bigger than those considered in the exaustive case (the average execution time for $\card{A}=32$ and $\card{B}=2$ is around $4500$ ms).
Analogously to the exhaustive algorithm we have reported some slices of Figure~\ref{fig:SumCompExistGlobal}  in Figure~\ref{fig:SumCompExistGlobalSliceB2} (by setting $\card{B}=2$) and in Figure~\ref{fig:SumCompExistGlobalSliceA32} (by setting $\card{A}=32$) by considering different intervals where the values of $A$ are chosen.
The results are analogous with those obtained in the exhaustive algorithm, but we have to observe that the increase in time execution is obtained for higher values of $A$ ($\card{A}>30$) and that we obtain higher execution times  for the range $[1-200]$ and for the case with $\card{B}=2$.

Figure~\ref{fig:SumCompExistA32DistinctRange} shows the execution times at increasing number of duplicated values and by considering different ranges of values for $A$. Also in this case we can remark that our representation of partitions has positive effects in the performances of the algorithm.  This is more evident for the ranges $[1,100]$ and $[1,200]$, whereas in the range $[1,150]$ we have an outlier (average execution time $4,219$ ms in the case of number of repetitions of elements in $A$ equal to $8$) that has been confirmed by other experiments that we have conducted on the same range of values.


Figure~\ref{fig:SumCompExistNotFound} reports the average execution time of the algorithm applied to cases where an enabling de\-com\-pos\-ition  does not exists. These cases have been generated by rule $R2$ discussed at the beginning of the section and considering different lengths of $A$ and $B$ ($\card{A}$ from $2$ to $32$ elements). It is easy to see that on average the execution time is lesser than 1 ms. By comparing these results with those reported in Figure~\ref{fig:SumCompExistGlobal} (where the enabling de\-com\-pos\-itions always exist) we can observe that identifying the lack of sum composition is much faster than identifying its presence of several order of magnitude. This result is the opposite of what we were expecting. We performed a deeper analysis on these results and we noticed that, for the cases of non existence of sum composition, the \emph{simplification} mechanisms were very effective by significantly reducing  the length of the partitions $A$ and $B$ and, consequently,  reducing the execution time of the algorithm

\begin{figure}
  \centering
  \begin{minipage}[b]{0.45\textwidth}
\includegraphics[scale=.3]{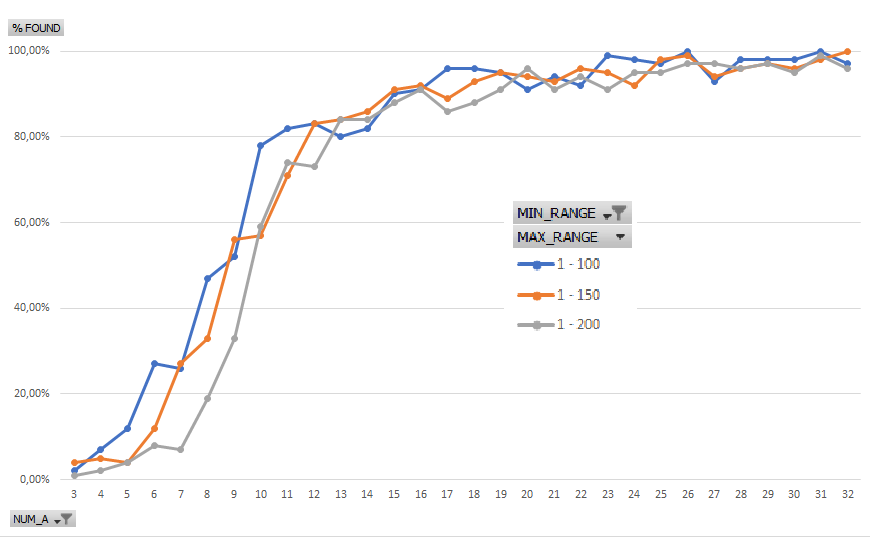}
\caption{Frequency of existence of sum composition w.r.t. the generated datasets by fixing $\card{B}=2$  and by varying the length of $A$ for all ranges  (with $3 \leq \card{A} \leq 32$)} \label{fig:SumCompExistBasePerFoundB2}   
  \end{minipage}
  \ \hspace*{20pt} \
  \begin{minipage}[b]{0.45\textwidth}
\includegraphics[scale=.3]{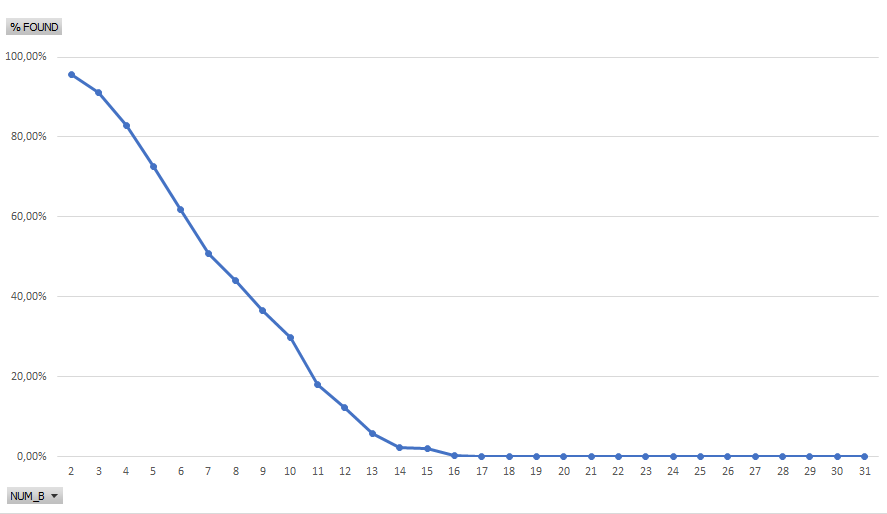}
\caption{Frequency of lacking of sum composition w.r.t. the generated datasets by fixing $\card{A}=32$ for all ranges and by varying the length of $B$  (with $2 \leq \card{B} \leq 31$)} \label{fig:SumCompExistBasePerFoundA32}
  \end{minipage}
\end{figure}

Figure~\ref{fig:SumCompExistBasePerFound} reports the frequency of the occurrence of enabling de\-com\-pos\-itions w.r.t. the number of cases that have been randomly generated in our datasets. 
The frequency is much higher when the length of $B$ is low (around $2$) and the length of $A$ is high.
The trend is better represented in Figure~\ref{fig:SumCompExistBasePerFoundB2} where a slice of the previous figure is shown by fixing $\card{B}=2$ and in Figure~\ref{fig:SumCompExistBasePerFoundA32} where a  slice of previous figure is shown by fixing $\card{A}=32$. From these graphics we can note that the frequency with which enabling de\-com\-pos\-itions are identified is $16,56\%$. Moreover, when the ratio $\card{A}/\card{B}$ is low, the frequency of cases in which enabling de\-com\-pos\-itions are not identified rises to $100\%$. The highest numbers of these cases is verified when $\card{A}/\card{B}=2$ and this result can be justified by  Theorem~\ref{theo:3}. For higher value of this ratio, it could be worth investigating  if the property holds also for higher numbers of $A$.


The last experiments were devoted to evaluate the impact of the \emph{simplification} properties introduced in Section~\ref{subsection:existence} on the performances of the existential algorithm.  For this purpose,  two kinds of \emph{simplifications} are considered:
\begin{itemize}[nosep]
\item A \emph{full simplification}: when the $\textsc{sumCompExist}$ algorithm can decide if the sum composition holds or not without the need to call the {\sc SumCompExistAux} function. This happens when one of the conditions of the statements 1, 2, 3, 5, 6 and 8 are met in Algorithm~\ref{alg2}.
\item A \emph{partial simplification}: when  $\textsc{sumCompExist}$ can reduce the execution time of the algorithm by decreasing the numbers of elements of $A$ and $B$ before calling the {\sc SumCompExistAux} function. This happens when the statements at line 4 or at line 7 is applied. 
\end{itemize}

\begin{figure}[t]
\centering
\subfigure[]{\label{fig:SumCompExistSimplTotal}\includegraphics[scale=.29]{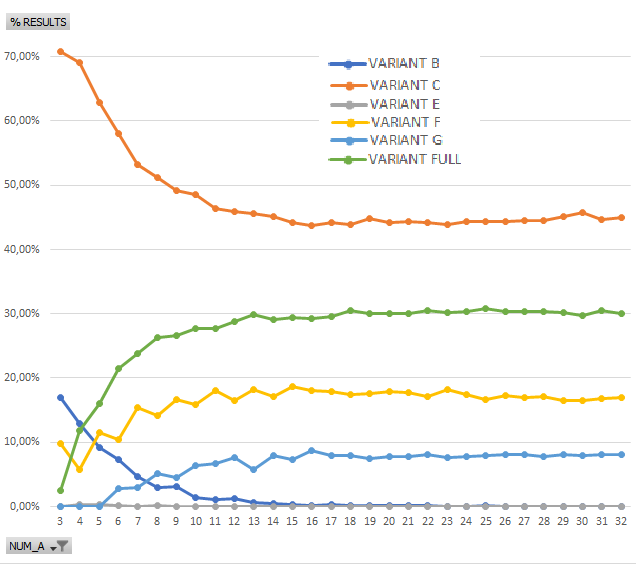}}  \ \hspace*{6pt} \
\subfigure[]{\label{fig:SumCompExistSimplTime}\includegraphics[scale=.29]{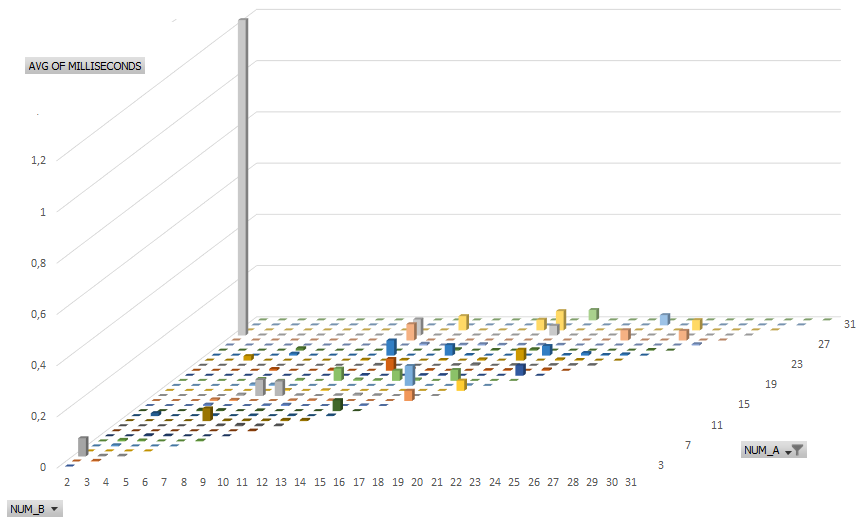}} \\
\subfigure[]{\label{fig:SumCompExistSimplRatio}\includegraphics[scale=.29]{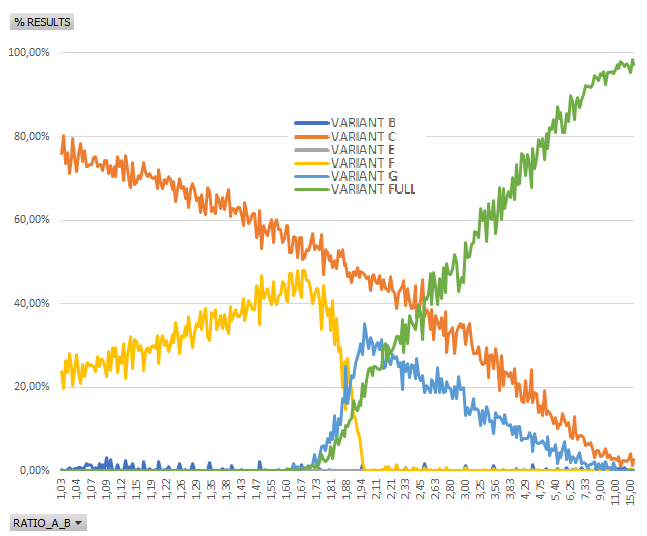}} \  \hspace*{6pt} \
\subfigure[]{\label{fig:SumCompExistSimplA}\includegraphics[scale=.29]{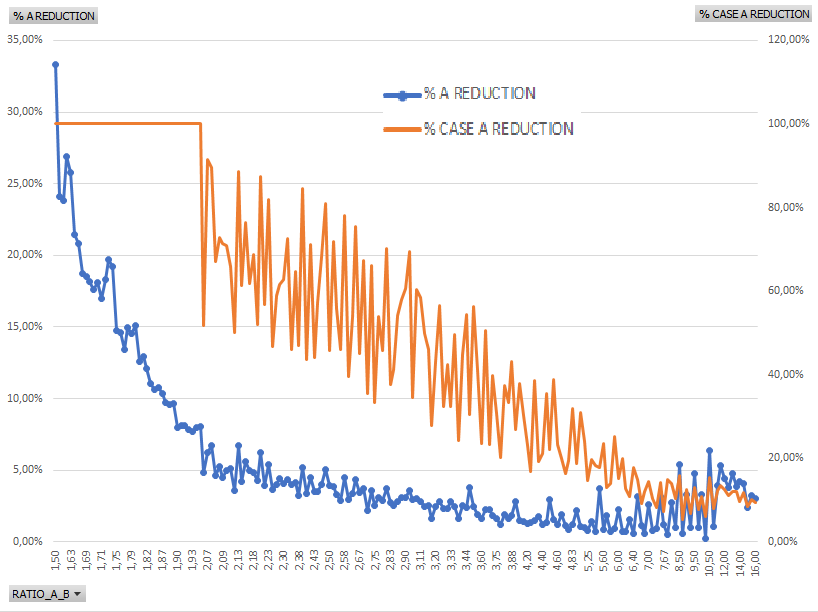}}
\caption{Effect of simplifications on function $\textsc{sumCompExist}$. (a) Percentage of \emph{full simplification} variants. (b) average execution time  for \emph{full simplified} variants. (c) Percentage of \emph{full simplification} variants in relation with the  ratio $\card{A}/\card{B}$. (d) Percentage of \emph{partial simplification} on $A$ in relation with the  ratio $\card{A}/\card{B}$}\label{fig:XXX}
\end{figure}

For the \emph{full simplification}, the following  situations of exit from Algorithm~\ref{alg2} are analyzed:
\begin{itemize}[nosep]
\item {\em variant B}. Exit at line  2. 
\item {\em variant C}. Exit at line  3. 
\item {\em variant E}. Exit at line  5. 
\item {\em variant F}. Exit at line  6. 
\item {\em variant G}. Exit at line  8. 
\end{itemize}

These variants are compared with the entire execution of the Algorithm~\ref{alg2} which is denoted {\em variant FULL}.
Note that the exit at line 1 from Algorithm~\ref{alg2} is not considered because this situation cannot ever be verified because of the use of rule $R2$ in the generation of the datasets.

Figure~\ref{fig:SumCompExistSimplTotal} reports the percentage of cases for the different variants to Algorithm~\ref{alg2}. We can see that, for the ranges of values observed, we have that the \emph{full simplification} can provide the existence or not existence of sum composition for roughly 70\% of the cases. This result is quite promising, because despite the high complexity of sum composition, it seems that for a significant percentage of cases, the algorithm can provide an answer quite quickly. 
The figure shows  that the maximum average time  is 1.23 ms ($\card{B}=2$ and $\card{A}=29$) that is 3 order of magnitude smaller than the case of not \emph{full simplification} (4,469 ms when $\card{B}=2$ and $\card{A}=32$ in Figure~\ref{fig:SumCompExistGlobal}). Another interesting behavior can be observed  in Figure~\ref{fig:SumCompExistSimplTotal}: when $\card{A} \geq 16$ the percentage of \emph{full simplification} becomes stable for all the cases. This fact deserves further  study  in the future.

Figure~\ref{fig:SumCompExistSimplTime} shows the execution times of the different variants w.r.t. the ratio of $\card{A}$ and $\card{B}$. We can observe that till the ratio is $3$, this \emph{simplification} is quite effective in at least 50\% of the cases. When this ratio increases, the efficacy of the simplification decreases rapidly. It would be interesting to verify with higher length of $A$ if these results are confirmed. Another  investigation direction  is to identify other kinds of \emph{simplification} that can be applied.

Figure~\ref{fig:SumCompExistSimplA} shows the effects of \emph{partial simplification} (only for the call of $\textsc{SumCompExistAux}$) of the partition $A$ by applying  the statements at line 4 and 7 in the $\textsc{sumCompExist}$ algorithm. In particular the
line ``{\small \% CASE A REDUCTION}"
 is the percentage of the cases where we can reduce the number of elements of $A$,
line ``{\small \% A REDUCTION}"
 is the ratio (number of $A$ eliminated)/$\card{A}$ (two different scales used). The x-axis is the ratio $\card{A}/\card{B}$ because it is the most reasonable due the variation of $A$ and $B$ on our tests. We can observe  that this kind of simplification is applied in all the cases when $\card{A}/\card{B} \le 2$ then
it rapidly decreases
arriving at an application of this \emph{simplification} at roughly 10\% of the cases when $\card{A}/\card{B} \ge 10$. The percentage of number of $A$ eliminated is even less with a decreasing curve until $\card{A}/\card{B} = 10$ and then a strange small increase until $\card{A}/\card{B} = 16$. This kind of \emph{simplification} can be considered interesting but marginal w.r.t. the most complex cases.

\section{Related Work}\label{sec:rw}
The sum composition problem, as defined in this article, is hard to find in the literature.
A similar class of problems, named  {\em Optimal Multi-Way Number Partitioning}, can be found in \cite{SEK18}. In their case, the authors provide different optimal  algorithms for separating a partition $A$ of $n$ positive integers into $k$ subsets such that the largest sum of the integers assigned to any subset is minimized.  The paper provides an overview of different algorithms that fall into three categories: sequential number partitioning (SNP); binary-search improved bin completion (BSIBC); and, cached iterative weakening (CIW). They show experimentally that for large random numbers, SNP and CIW outperform state of the art by up of seven orders of magnitude in terms of runtime.
The problem is slightly different from the one that we face in this paper in which we wish to partition the set $A$ according to the values contained in the set $B$. However, we also exploit a SNP approach for the generation of the possible partitions of the partition $A$ according to the integer numbers contained in $B$. Peculiarity of our approach is the representation of the partition $A$ and $B$ that allows one to reduce the cases to be explored by eliminating useless permutations. 

A much more similar formulation of the problem faced in this paper is the ``k-partition problem'' proposed in  \cite{Chaimovich1993AFP} where the authors wish to minimize a objective function $z$ with the partitions $A$ and $B$ (with the sum of $A$ equals the sum of $B$) where $A$ is partitioned in $m$ subpartitions $(A_1, \ldots, A_m)$. For each partition the  function $z$ is calculated as $z=max_{i=1}^{m}(\sum_{a \in A_i}a)/b_i$.
It is easy to see that if the minimum of objective function is $z=1$ then $B$ is sum composition of $A$ and, vice versa, if $B$ is sum composition of $A$, the objective function $z=1$. In \cite{Chaimovich1993AFP}, special restrictions are introduced concerning the length and distinct values of $A$, the maximal integer in $A$ and the minimal integer in $B$ in order to provide more efficient algorithms.
 For instance, if $\card{B}=2$  and $max(A)=100$, then  $\cardd{A}>94$ (i.e. almost all the value from 1 to 100) and $min(B)>2366$.
 Even if there are many similarities with the approach proposed in this paper, in \cite{Chaimovich1993AFP} no real implementation of the algorithms are proposed and the paper lacks of experimental results.
 Moreover, the paper provides a response to the existence of a solution for the $k$-partition problem but does not provide algorithms for the identification of all possible solutions as we propose in this paper. Finally, our paper provides a characterization of the properties of sum composition that are exploited for reducing the number of cases to be tested. Even if our approach is still NP-hard, its runtime is significantly reduced especially when checking the existence of a solution, and we do not provide any restrictions on the partitions $A$ and $B$.

In our algorithms  we need to go  through the identification of all the solutions of the well-known ``Subset Sum'' Problem that is one of Karp's original NP-complete problems \cite{Karp1972}. This problem consists in determining the existence of a subset $A'$ of $A$ whose sum is $s$ and is  a well-known example of a problem that can be solved in
weakly polynomial time. As a weakly NP-complete problem, there is a standard pseudopolynomial time algorithm using a dynamic
programming, due to Bellman, that solves it in ${\mathcal O}(\card{A} \sigma(A'))$ time \cite{Bell58}. Recently some better solutions have need proposed by Koiliaris and Xuy in \cite{Koiliaris2017AFP} with an algorithms that runs in $\tilde{{\mathcal O}}(min\{\sqrt{\cardd{A}} \sigma(A'), \sigma(A')^{\frac{4}{3}}, \sigma(A)\})$ time, where $\tilde{{\mathcal O}}$
hides polylogarithmic factors.
We remark that in our case, it is not enough to determine if a subset exists but we need to identify all the possible subsets in $A$ whose sum is $s$. Therefore, the standard approaches proposed in the literature cannnot be directly applied. In our case,
we keep the values of $A$ ordered and we  apply a SNP approach for enumerating all the possible subpartition of $A$ (starting from the lower values of $A$) that can lead to the sum $s$. A subpartition is skipped when including a new value of $A$, this leads to a value greater that $s$. Since our partitions are ordered, we can guarantee to identify a possible solution, whereas the approach proposed in \cite{Koiliaris2017AFP} randomly generates possible configurations to be proved.
Finally, relying on our representation of partitions, permutations are avoided and  the configurations to test are reduced.

\section{Conclusions}\label{sec:conclusion}

In this paper we introduced the sum composition problem between two partitions $A$ and $B$ of positive integers. Starting from a formal presentation of the problem by exploiting the poset theory, we have proposed several properties  that can be exploited for improving the execution time of the developed algorithms.
Then, we have developed an exhaustive algorithm
for the generation of all $AB$-de\-com\-pos\-itions for which $ A \leq B $,
and an  algorithm for checking the existence of the relation $ A \leq B $.
The correctness of these two algorithms is proved.

An experimental analysis is provided for assessing the quality of the proposed solutions.
As expected, the algorithms have an exponential growth with the length of $A$ and $B$.
We also show a correlation between the execution time and the number of enabling de\-com\-pos\-itions
in the execution of the exhaustive algorithm.
Moreover, we show that the number of repetitions of the elements in the partition $A$
impacts the execution time for both algorithms and how the adopted data structures reduce the number of configurations to be checked.  For the existential algorithm, we had some surprising good impact on ``full simplification'' actions on roughly $70\%$ of the cases and also ``partial simplification'' can have good impact in the range of $\card{A}/\card{B} \leq 3$. Other interesting experimental evidences were identified as the ratio of $\card{A}/\card{B}$ for the enabling de\-com\-pos\-itions (exhaustive algorithm). Anyway, these algorithms present a limitation of  ``lack of memory'' ($\card{A}=24$ for the exhaustive algorithm and $\card{A}=33$ for the existence algorithm).

New area of investigations can be identified in finding:
$i)$ new \emph{simplification} properties that can reduce the execution time of the algorithms;
$ii)$ new algorithms that have less limitation of memory usage;
$iii)$  a wider range of test cases where the validity of certain properties can be checked.

\small

\end{document}